%% file: eilenberg_nom.tex
\documentclass[a4paper,british,numberwithinsect,autoref,final]{lipics-v2019}
\nolinenumbers
\usepackage[utf8]{inputenc}
\usepackage[T1]{fontenc}

  \usepackage{lipsum}

  \newcounter{mycounter}

\include{preamble_nom}

\addto\extrasbritish{
}

\usepackage[nomargin,marginclue,footnote]{fixme}
\FXRegisterAuthor{sm}{asm}{SM}
\FXRegisterAuthor{hu}{ahu}{HU}

\usepackage{cite}
\usepackage{microtype}
\def\o{\cdot}
\title{Varieties of Data Languages}
\titlerunning{Varieties of Data Languages}

\author{Henning Urbat}{Friedrich-Alexander-Universität Erlangen-Nürnberg, Germany}{henning.urbat@fau.de}{}{}

\author{Stefan Milius}{Friedrich-Alexander-Universität Erlangen-Nürnberg, Germany}{mail@stefan-milius.eu}{}{}

\authorrunning{H.~Urbat and S.~Milius}

\Copyright{H.~Urbat and S.~Milius}

\funding{Stefan Milius acknowledges support by Deutsche Forschungsgemeinschaft (DFG) under project MI~717/5-1. Henning Urbat acknowledges support by Deutsche Forschungsgemeinschaft (DFG) under project SCHR~1118/8-2}

\EventEditors{Christel Baier, Ioannis Chatzigiannakis, Paola Flocchini, and Stefano Leonardi}
\EventNoEds{4}
\EventLongTitle{46th International Colloquium on Automata, Languages, and Programming (ICALP 2019)}
\EventShortTitle{ICALP 2019}
\EventAcronym{ICALP}
\EventYear{2019}
\EventDate{July 9--12, 2019}
\EventLocation{Patras, Greece}
\EventLogo{eatcs}
\SeriesVolume{132}
\ArticleNo{125}

\begin{document}
\maketitle
\begin{abstract}
  We establish an Eilenberg-type correspondence for data languages,
  i.e.~languages over an infinite alphabet. More precisely, we prove
  that there is a bijective correspondence between varieties of
  languages recognized by orbit-finite nominal monoids and
  pseudovarieties of such monoids. This is the first result of this
  kind for data languages.  Our approach makes use of nominal Stone
  duality and a recent category theoretic generalization of
  Birkhoff-type theorems that we instantiate here for the category
  of nominal sets. In addition, we prove an axiomatic characterization
  of weak pseudovarieties as those classes of orbit-finite monoids
  that can be specified by sequences of nominal equations, which
  provides a nominal version of a classical theorem of Eilenberg
  and Schützenberger.
\end{abstract}

\section{Introduction}

In the algebraic theory of formal languages, one studies
automata and the languages they represent in terms of associated
algebraic structures. This approach has been successfully implemented
for numerous types of languages and has proven extremely fruitful because
it allows to import powerful algebraic methods into the realm of
automata theory. As a prime example, regular languages can be
described purely algebraically as the languages recognized by finite
monoids, and a celebrated result by McNaughton, Papert, and
Schützenberger~\cite{mp71,sch65} asserts that a regular language is definable in
first-order logic if and only if its syntactic monoid is aperiodic
(i.e.~it satisfies the equation $x^{n+1}=x^n$ for sufficiently large
$n$). As an immediate application, this algebraic characterization
yields an effective procedure for deciding first-order
definability. The first systematic approach to correspondence results
of this kind was initiated by Eilenberg~\cite{eilenberg76} who proved
that \emph{varieties of languages} (i.e.~classes of regular languages
closed under the set-theoretic boolean operations, derivatives, and
homomorphic preimages) correspond bijectively to \emph{pseudovarieties
  of monoids} (i.e.~classes of finite monoids closed under quotients monoids,
submonoids, and finite products). Eilenberg's result thus establishes
a {generic} relation between properties of regular languages and
properties of finite monoids. In addition, Eilenberg and Schützenberger~\cite{es76}
contributed a model-theoretic description of pseudovarieties:
they are those classes of finite monoids that can be axiomatized by a
sequence $(s_n=t_n)_{n \in \Nat}$ of equations,
interpreted as ``$s_n=t_n$ holds for sufficiently large $n$''. For
instance, the pseudovariety of aperiodic finite monoids is axiomatized
by $(x^{n+1}=x^n)_{n \in \Nat}$.

The goal of our present paper is to study \emph{data languages},
i.e.~languages over an infinite alphabet, from the perspective of
algebraic language theory. Such languages have spurred significant
interest in recent years, driven by practical applications in various
areas of computer science, including efficient parsing of XML
documents or software verification. Mathematically, data languages are
modeled using \emph{nominal sets}. Over the years, several machine
models for handling data languages of different expressive power have
been proposed; see~\cite{seg06,sch07} for a comprehensive survey. The
focus of this paper is on languages recognized by \emph{orbit-finite nominal
monoids}. They form an important subclass of the languages accepted
by Francez and Kaminski's \emph{finite memory automata}~\cite{fk94}
(which are expressively equivalent to orbit-finite automata in the
category of nominal sets~\cite{bkl14}) and have been characterized by
a fragment of monadic second-order logic over data words called
\emph{rigidly guarded MSO}~\cite{cpl15}. In addition, Boja\'nczyk
\cite{boj2013} and Colcombet, Ley, and Puppis~\cite{cpl15} established
nominal versions of the McNaughton-Papert-Schützenberger theorem and
showed that the first-order definable data languages are precisely the ones
recognizable by aperiodic orbit-finite monoids.

In the light of these results, it is natural to ask whether a generic variety
theory akin to Eilenberg's seminal work can be developed for data
languages. As the main contribution of our paper, we answer this
positively by establishing nominal generalizations of two key results
known from the algebraic theory of regular languages. The first one is
a counterpart of Eilenberg's variety theorem, which is the
first result of this kind for data languages:
\vspace*{-5pt}
\begin{nomeilthm}
  Varieties of data languages correspond bijectively to
  pseudovarieties of nominal monoids.
\end{nomeilthm}
\vspace*{-5pt}

\noindent
Here, the notion of a \emph{pseudovariety of nominal monoids} is as
expected: a class of orbit-finite nominal monoids closed under
quotient monoids, submonoids, and finite products. In contrast, the notion of
a \emph{variety of data languages} requires two extra conditions
unfamiliar from other Eilenberg-type correspondences,
most notably a technical condition called \emph{completeness}
(\autoref{def:varlang}). Like the original Eilenberg theorem, its nominal version gives rise to a generic relation between properties of
data languages and properties of nominal monoids. For instance,
the aperiodic orbit-finite monoids form a pseudovariety, and the first-order
definable data languages form a variety, and thus the equivalence of
these concepts can be understood as an instance of the nominal
Eilenberg correspondence.

On a conceptual level, our results crucially make
use of \emph{duality}, specifically an extension of Petri\c{s}an's~\cite{petrisan12} nominal version of Stone duality which gives a dual equivalence between nominal sets and nominal complete atomic boolean algebras. To derive the nominal Eilenberg correspondence, we make two
key observations. First, we show that varieties of data languages dualize
(under nominal Stone duality) to the concept of an \emph{equational
  theory} in the category of nominal sets. Second, we apply a recent
categorical generalization of Birkhoff-type variety theorems~\cite{mu19} to
show that equational theories correspond to pseudovarieties of nominal
monoids. Our approach is summarized by the diagram below:

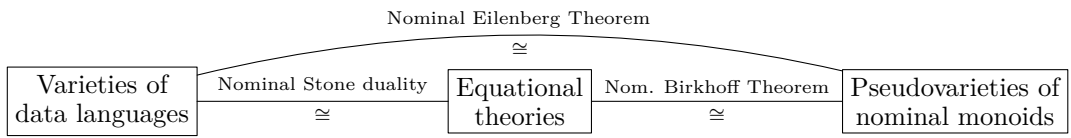
\vspace*{-15pt}
\[
  \xymatrix@C+7em{
    *+[F]{\txt{Varieties of\\ data languages}}
    \ar@{-}@/^2.5em/[rr]^-{\text{Nominal Eilenberg Theorem}}_-\cong
    \ar@{-}[r]^-{\text{Nominal Stone duality}}_-\cong
    &
    *+[F]{\txt{Equational\\theories}}
    \ar@{-}[r]^-{\text{Nom. Birkhoff Theorem}}_-\cong
    &
    *+[F]{\txt{Pseudovarieties of \\ nominal monoids}}
  }
\]
The idea that Stone-type dualities
play a major role in algebraic language theory was firmly
established by Gehrke, Grigorieff, and Pin~\cite{ggp08}. It is also at the
heart of our recent line of work~\cite{ammu14,ammu15,uacm17,ammu19}, which
culminated in a uniform category theoretic proof of more than a dozen
 Eilenberg correspondences for various types of languages. A related, yet more abstract, approach was pursued by Salam\'anca \cite{s16}. The key insight of \cite{uacm17,s16} is that Eilenberg-type correspondences arise by combining a Birkhoff-type correspondence with a Stone-type duality. Our present approach to data languages is an implementation of this principle in the nominal setting. Since the existing categorical frameworks for algebraic language theory consider algebraic-like base categories (which excludes nominal sets) and the recognition of languages by finite structures, our Nominal Eilenberg Theorem is not covered by any previous categorical work and requires new techniques. However, our approach can be seen as an indication of the robustness of the duality-based methodology for algebraic recognition.

As our second main contribution, we complement the Nominal Eilenberg Theorem with a model-theoretic
description of pseudovarieties of nominal monoids in terms of sequences of \emph{nominal
  equations}, generalizing the classical result of Eilenberg and
Schützenberger for ordinary monoids. Our result applies more
generally to the class of \emph{weak pseudovarieties} of nominal
monoids, which are only required to be closed under support-reflecting
(rather than arbitrary) quotients. We then obtain the
\vspace*{-5pt}
\begin{nomeilschuetz}
  Weak pseudovarieties are exactly the classes of nominal monoids
  axiomatizable by sequences of nominal equations.
\end{nomeilschuetz}
\vspace*{-3pt}

While our main results apply to languages recognizable by
orbit-finite monoids, the underlying methods are of fairly general
nature and can be extended to other recognizing structures in the
category of nominal sets. We illustrate this in
\autoref{sec:regdata} by deriving a (local) Eilenberg correspondence
for languages accepted by deterministic nominal automata.


\section{Nominal Sets}\label{S:nom}
\label{S:prelim}
We start by recalling basic definitions and facts from the theory of
nominal sets~\cite{pitts2013}. Some of the concepts considered in this paper are most clearly and conveniently formulated in the
language of category theory, but only very basic knowledge of
category theory is required from the reader. Fix a countably infinite set $\At$ of \emph{atoms}, and
denote by $\Perm(\At)$ the group of finite permutations of $\At$
(i.e.~bijections $\pi\colon \At\to\At$ that move only finitely many
elements of $\At$). A \emph{$\Perm(\At)$-set} is a set $X$ with an operation $\Perm(\At)\times X\to X$, denoted as
$(\pi,x)\mapsto \pi\o x$, such that
$(\sigma \pi)\o x = \sigma\o (\pi \o x)$ and $\id \o x$ for all
$\sigma,\pi\in\Perm(\At)$ and $x\in X$. If the group action is trivial, i.e.~$\pi\o x= x$ for all $\pi\in \Perm(\At)$ and $x\in X$, we call $X$ \emph{discrete}. For any set $S\seq \At$ of atoms, denote by
$\Perm_S(\At)\seq \Perm(\At)$ the subgroup of all finite permutations
$\pi$ that fix $S$, i.e.~$\pi(a)=a$ for all $a\in S$. The set
$S$ is called a \emph{support} of an element $x\in X$ if for every
 $\pi\in \Perm_S(\At)$ one has $\pi\o x = x$. The intuition is that $x$ is some kind of syntactic object (e.g.~a string, a tree, a term, or a program) whose free variables are contained in $S$. Thus, a variable renaming $\pi$ that leaves $S$ fixed does not affect $x$. A
\emph{nominal set} is a $\Perm(\At)$-set $X$ such that every element
of $X$ has a finite support. This implies that every element $x\in X$ has
a least support, denoted by $\supp_X(x)\seq \At$. A nominal set $X$ is \emph{strong} if, for every $x\in X$ and $\pi\in\Perm(\At)$, one has
$\pi\o x = x$ if and only if $\pi(a)=a$ for all $a\in \mathsf{supp}_X(x)$. The \emph{orbit} of an element $x$ of a nominal set $X$
is the set $\{\pi\o x\;:\; \pi\in \Perm(\At)\}$. The orbits form a partition of
$X$. If $X$ has only finitely many orbits, then $X$ is called
\emph{orbit-finite}. More generally, for any finite set $S\seq \At$ of atoms, the \emph{$S$-orbit} of an element $x\in X$ is the set $\{\pi\o x \;:\; \pi\in \Perm_S(\At)\}$, and the $S$-orbits form a partition of $X$.

\begin{lemma}\label{lem:s_orbit_finite}
  Let $S$ be a finite subset of $\At$. Then every orbit-finite nominal
  set has only finitely many $S$-orbits.
\end{lemma}
A map $f\colon X\to Y$ between nominal sets is \emph{equivariant} if $f(\pi\o x)=\pi\o f(x)$ for all $\pi\in\Perm(\At)$ and $x\in X$, and \emph{finitely supported} if there exists a finite set $S\seq \At$ such that $f(\pi\o x)=\pi\o f(x)$ for all $\pi\in\Perm_S(\At)$ and $x\in X$.  Equivariant maps do not increase supports, i.e.~one has $\supp_Y(f(x))\seq \supp_X(x)$ for all $x\in X$. We write $\Nomfs$ for the category of nominal sets and
  finitely supported maps, and $\Nom$ for the (non-full) subcategory of nominal sets and equivariant maps. We shall use the following standard results about $\Nom$:
\begin{enumerate}
\item $\Nom$ is complete and cocomplete. Finite limits and all
  colimits are formed on the level of underlying sets. In particular,
  finite products of nominal sets are given by cartesian products and
  coproducts by disjoint union.


\item For every pair $X,Y$ of nominal sets, the \emph{exponential}
  $[X,Y]$ is the nominal set  consisting of all finitely supported maps
  $f\colon X\to Y$, with the group action given by
  $(\pi \o f)(x) = \pi\o f(\pi^{-1}\o x)$.  Moreover, for every
  nominal set $X$, the \emph{nominal power set} $\Pow X$ 
  is carried by the set of all subsets $X_0\seq X$ with finite
  support; i.e.~for which there exists a finite set $S\seq \At$ of
  atoms such that $\pi \o X_0 = X_0$ for $\pi\in\Perm_S(\At)$, where
  $\pi\o X_0 = \{\pi\o x\;:\; x\in X_0\}$. In particular, every
  singleton $\set{x}$ is finitely supported by $\supp_X(x)$. The group
  action on $\Pow X$ is given by $X_0\mapsto \pi\o X_0$, and we have
  $\Pow X \cong [X,2]$, for the discrete nominal set $2 = \{0,1\}$.
  
\item \emph{Quotients} and \emph{subobjects} in $\Nom$ are represented by epimorphisms (= surjective equivariant maps) and monomorphisms (= injective equivariant maps), respectively. $\Nom$ has image factorizations, i.e.~every
  equivariant map $f\colon X \to Y$ has a unique decomposition
  $f = m \cdot e$ into a quotient
  $e\colon X \epito I$  followed by a subobject
  $m\colon I \monoto Y$. We call $e$ the \emph{coimage} of $f$.
  
\item Orbit-finite nominal sets are closed under quotients, subobjects, and finite products.

\item For each $n\geq 0$, the nominal set
  $\At^{\# n} = \{\,(a_1,\ldots,a_n)\in \At^n \;:\; a_i\neq a_j \text{
    for $i\neq j$}\,\}$ with group action
  $\pi\o (a_1,\ldots, a_n) = (\pi(a_1),\ldots, \pi(a_n))$ is strong
  and has a single orbit. More generally, the (orbit-finite) strong
  nominal sets are up to isomorphism exactly the (finite) coproducts
  of nominal sets of the form $\At^{\#n}$.

%
\end{enumerate}

\section{Pseudovarieties of Nominal Monoids}
In this section, we investigate classes of orbit-finite nominal monoids and establish two characterizations of such classes: a categorical one, relating pseudovarieties of nominal monoids to {equational theories} in the category of nominal sets, and an axiomatic one, describing weak pseudovarieties in terms of sequences of nominal equations. The first of these results is the algebraic foundation of our subsequent treatment of varieties of data languages.

A \emph{nominal monoid} is a monoid $(M,\bullet,1_M)$ in the category
$\Nom$; that is, $M$ is equipped with the structure of a nominal set
such that the multiplication $\bullet\colon M\times M\to M$ is an
equivariant map and the unit $1_M\in M$ has empty support, i.e.~it
corresponds to an equivariant map $1\to M$, where $1$ is the nominal
set with one element. We write $\NomMon$ for the category of nominal
monoids and equivariant monoid morphisms (usually just called
\emph{morphisms}), and $\NomMon_\of$ for the full subcategory of
orbit-finite nominal monoids. The forgetful functor from $\NomMon$ to
$\Nom$ has a left adjoint assigning to each nominal set $\Sigma$ the
free nominal monoid $\Sigma^*$ of all words over $\Sigma$, with monoid
multiplication given by concatentation of words, unit $\epsilon$ (the
empty word) and group action
$\pi\o (a_1\cdots a_n)=\pi(a_1)\cdots \pi(a_n)$ for
$\pi\in \Perm(\At)$ and $a_1\cdots a_n\in \Sigma^*$. The category
$\NomMon$ has products (formed on the level of $\Nom$), image
factorizations, and surjective and injective morphisms
represent \emph{quotients} and \emph{submonoids} of nominal
monoids. A quotient $q\colon M \epito M'$ is called
\emph{support-reflecting} if for every $x' \in M'$ there exists an
$x \in M$ with $q(x) = x'$ and $\supp_M(x) = \supp_{M'}(x')$. The
following result characterizes the quotient monoids of $\Sigma^*$ in
terms of unary operations:
\begin{proposition}[Unary presentation for nominal monoids]\label{prop:unarypres}
  For every nominal set $\Sigma$ and every surjective equivariant map
  $e\colon \Sigma^* \epito M$, the following statements are
  equivalent:
  \begin{enumerate}
  \item\label{prop:unarypres:1} $e$ carries a quotient monoid of
    $\Sigma^*$, i.e.~there exists a nominal monoid structure
    $(M,\bullet,1_M)$ on $M$ such that
    $e\colon \Sigma^*\epito (M,\bullet,1_M)$ is a morphism of nominal
    monoids;
  \item\label{prop:unarypres:2} the maps
    $\Sigma^*\xra{w\o \dash} \Sigma^*$ and
    $\Sigma^*\xra{\dash\o w} \Sigma^*$ $(w\in \Sigma^*)$ lift along
    $e$, i.e.~there exist (necessarily unique) maps $l_w$ and $r_w$ making the
    following squares commute:
    \[
      \vcenter{
        \xymatrix@R-0.6em{
          \Sigma^* \ar@{->>}[d]_e \ar[r]^-{w\o \dash} & \Sigma^* \ar@{->>}[d]^e \\
          M \ar@{-->}[r]_-{l_w} & M
        }}
      \qquad\qquad
      \vcenter{
        \xymatrix@R-0.6em{
          \Sigma^* \ar@{->>}[d]_e \ar[r]^-{\dash\o w} & \Sigma^* \ar@{->>}[d]^e \\
          M \ar@{-->}[r]_-{r_w} & M
        }}
      \qquad\qquad
      \text{for every $w \in \Sigma^*$}.
    \]
  \end{enumerate}
\end{proposition}
In general, the maps $w\o \dash$ and $\dash\o w$ are not equivariant, but finitely supported (with support
contained in the one of $w$). This implies that also $l_w$ and
$r_w$ in~\ref{prop:unarypres:2} are finitely supported.

\subsection{Equational Theories}

In previous work~\cite{mu19} we studied varieties of objects in a general category and their relation to an abstract form of equations. In the following, we instantiate these concepts to the category of nominal sets to derive a characterization of pseudovarieties of orbit-finite monoids.
\begin{defn}
  Let $\Sigma$ be a nominal set. A \emph{$\Sigma$-generated nominal
    monoid} is a nominal quotient monoid $e\colon \Sigma^*\epito M$ of the free monoid
  $\Sigma^*$. We denote by $\Sigma^*\mathord{\epidownarrow} \NomMon_\of$ the poset of  $\Sigma$-generated orbit-finite nominal monoids, ordered by $e\leq e'$ iff $e'$
  factorizes through $e$.
\end{defn}
\begin{defn}
  A \emph{local pseudovariety of $\Sigma$-generated nominal monoids}
  is a filter
  $\T_\Sigma\seq \Sigma^*\mathord{\epidownarrow} \NomMon_\of$ in the
  poset of $\Sigma$-generated orbit-finite nominal monoids; that is,
  $\T_\Sigma$ is
  \begin{enumerate}
  \item \emph{upwards closed:} $e\in \T_\Sigma$ and $e\leq e'$ implies
    $e'\in \T_\Sigma$, and
  \item \emph{downwards directed:} for each pair $e_0,e_1\in \T_\Sigma$
    there exists $e\in \T_\Sigma$ with $e\leq e_0,e_1$.
  \end{enumerate}
  If we replace (1) by the weaker condition
  \begin{enumerate}[label=(\arabic*')]
  \item for each $e\colon \Sigma^*\epito M$ in $\T_\Sigma$ and each
    support-reflecting $q\colon M\epito N$ one has
    $q\o e \in \T_\Sigma$,
  \end{enumerate}
  then $\T_\Sigma$ is called a \emph{weak local pseudovariety of
    $\Sigma$-generated nominal monoids}.
\end{defn}
%
%
%
%

\begin{rem}\label{rem:localpseudovar_unpres}
  By \autoref{prop:unarypres}, the definition of local pseudovariety
  can be equivalently stated as follows:
  \begin{enumerate}
  \item $\T_\Sigma$ is a filter in the poset of orbit-finite quotients of $\Sigma^*$ in $\Nom$;
  \item for every $e\in \T_\Sigma$ and $w \in \Sigma^*$, the unary operations
    $w\o \dash$ and $\dash\o w$ on $\Sigma^*$ lift along $e$.
  \end{enumerate}
\end{rem}
Let $\Nom_\ofs$ denote the full subcategory of $\Nom$ on orbit-finite strong nominal
sets.
\begin{defn}[Equational Theory]\label{D:eqnth}
  A \emph{(weak) equational theory} is a family
  \[
    \T = (\,\T_\Sigma\seq \Sigma^*\mathord{\epidownarrow} \NomMon_{\of}\,)_{\Sigma\in
      \Nom_{\ofs}}
  \]
  of (weak) local pseudovarieties with the following two properties
  (see the diagrams below):
  \begin{enumerate}
  \item\label{D:eqnth:1}\emph{Substitution invariance.} For each
    equivariant monoid
    morphism $h\colon \Delta^*\to \Sigma^*$ with $\Delta, \Sigma\in\Nom_{\ofs}$ and each
    $e_\Sigma\colon \Sigma^*\epito M_\Sigma$ in $\mathscr{T}_\Sigma$, the coimage
    $e_\Delta$ of $e_\Sigma\o h$ lies in $\mathscr{T}_\Delta$.
  \item\label{D:eqnth:2}\emph{Completeness.} For each $\Sigma\in \Nom_{\ofs}$ and each
    quotient $e\colon \Sigma^*\epito M_\Sigma$ in $\mathscr{T}_\Sigma$, there exists
     $\Delta\in\Nom_{\ofs}$ and a support-reflecting quotient
    $e_\Delta\colon \Delta^*\epito M_\Delta$ in $\mathscr{T}_\Delta$ with $M_\Delta=M_\Sigma$.
  \end{enumerate}
\vspace{0.2cm}
  \[
    \xymatrix@=19pt{
      \Delta^* \ar[r]^{\forall h} \ar@{->>}[d]_{ e_\Delta} & \Sigma^*
      \ar@{->>}[d]^{\forall e_\Sigma}\\
      M_\Delta \ar@{ >->}[r] & M_\Sigma
    }
    \qquad\qquad
    \xymatrix@=19pt{
      \Delta^* \ar@{.>>}[d]_{\exists e_\Delta} & \Sigma^*
      \ar@{->>}[d]^{\forall e_\Sigma}\\
      M_\Delta\ar@{=}[r] & M_\Sigma
    }
  \]
\end{defn}
\begin{rem}\label{rem:complete}
  \begin{enumerate}
  \item\label{rem:complete:1} Local pseudovarieties were previously called
    \emph{equations}~\cite{mu19}. In fact, in many instances of the
    framework in \emph{op.~cit.}, a filter of quotients can be represented
    as a single quotient of a free algebra on an object $\Sigma$, which in turn corresponds to a set of pairs of terms
    given by the kernel of the quotient, i.e.~to the usual syntactic concept of an
    equation.
\item\label{rem:complete:3} The restriction to \emph{strong} nominal sets $\Sigma$ as generators reflects that the latter are the ``free'' nominal sets  \cite{KP10}, a property crucial for the proof of \autoref{thm:theories_vs_pseudovars} below. More precisely, letting $\Pow_f \At$ denote the set of finite subsets of $\At$, the forgetful functor $U\colon \Nom \to \Set^{\Pow_f \At}$ mapping a nominal set $X$ to the presheaf
$S \mapsto \{\, x\in X\;:\; \supp_X(x)\seq S\,\}$ has a left adjoint
$F$, and strong nominal sets are exactly the nominal sets of the form
$FP$ for $P\in \Set^{\Pow_f\At}$.
    
  \item\label{rem:complete:2} The somewhat technical completeness property
    cannot be avoided, i.e.~a substitution-invariant family of local
    pseudovarieties is generally incomplete. Indeed, consider
    the family
    \[
      \T = (\,\T_\Sigma\seq \Sigma^*\mathord{\epidownarrow} \NomMon_{\of}\,)_{\Sigma\in
        \Nom_{\ofs}},
    \]
    where $\T_\Sigma$ consists of all $\Sigma$-generated orbit-finite nominal monoids
    $e\colon \Sigma^*\epito M$ such that $e$ maps each element of $\Sigma^*$ with a support
    of size $1$ to $1_M$.

    \smallskip
    \noindent To see that $\T_\Sigma$ is a filter, suppose that
    $e\colon \Sigma^*\epito M$ and $e'\colon \Sigma^*\epito M'$ are two
    quotients in $\T_\Sigma$. Form their subdirect product $q$, viz.~the
    coimage of the morphism
    $\langle e,e'\rangle\colon \Sigma^*\to M\times M'$. Each $w\in \Sigma^*$ with a support of size $1$ is mapped by $q$ to
    $(e(w), e'(w))=(1_M,1_{M'}) = 1_{M\times M'}$. Thus $q\in \T_\Sigma$
    and $q\leq e,e'$, i.e.~$\T_\Sigma$ is downwards directed. Clearly, $\T_\Sigma$ is also upwards closed.

    \smallskip \noindent For substitution invariance, let
    $h\colon \Delta^*\to \Sigma^*$ be a morphism and $e_\Sigma\colon \Sigma^*\epito M_\Sigma$ a
    quotient in $\T_\Sigma$. Then $e_\Sigma\o h$ maps each element with a
    support of size $1$ to $1_{M_\Sigma}$ since $e_\Sigma$ does and the equivariant
     map $h$ does not increase
    supports. Thus, the coimage of $e_\Sigma\o h$ lies in
    $\T_\Delta$.

    \smallskip\noindent Finally, we show that $\T$ is not
    complete. Fix an arbitrary orbit-finite nominal monoid $M$
    containing an element $m$ with $\under{\supp_M{m}} = 1$. Note that $m\neq 1_M$ because $1_M$ has empty
    support. Moreover, choose an orbit-finite strong nominal set
    $\Sigma$ such that all elements of $\Sigma$ have least support of
    size at least $2$, and $M$ can be expressed as a quotient
    $e\colon \Sigma^*\epito M$. (For instance, one may take
    $\Sigma=\coprod_{i<k} \At^{\#n}$ where $k$ is the number of orbits
    of $M$ and $n = \max \{ 2, \under{\supp_M(x)} : x\in M \}$.) Since
    all nonempty words in $\Sigma^*$ have a least support of size
    at least $2$, one has $e\in \T_\Sigma$. For every
    $\Delta\in \Nom_{\ofs}$ and every quotient
    $q\colon \Delta^*\epito M$ in $\T_\Delta$, the set
    $q^{-1}[\{m\}]\seq \Delta^*$ contains no element with least support of
    size $1$, since such elements are mapped by $q$ to $1_M\neq
    m$. Consequently, $q$ is not support-reflecting. This shows that
    $M$ is not the codomain of any support-reflecting quotient in
    $\T_\Delta$.
  \end{enumerate}
\end{rem}
\begin{defn}[Pseudovariety and Weak Pseudovariety]
  A \emph{pseudovariety of nominal monoids} is a nonempty class
  $\pvar$ of orbit-finite nominal monoids closed under finite products,
  submonoids, and quotient monoids. That is,
  \begin{enumerate}
  \item for each $M,N\in \pvar$ one has $M\times N\in \pvar$;
  \item for each $M\in \pvar$ and each nominal submonoid $N\monoto M$ one has $N\in \pvar$;
  \item for each $M\in \pvar$ and each nominal quotient
    monoid $M\epito N$ one has $N\in \pvar$.
  \end{enumerate}
  A \emph{weak pseudovariety of nominal monoids} is a
  nonempty class of orbit-finite nominal monoids closed under finite
  products, submonoids, and support-reflecting quotient monoids. 
\end{defn}
The following result is a special case of the Generalized Variety
Theorem~\cite[Theorem 3.15]{mu19}. It asserts that equational theories
and pseudovarieties are equivalent concepts.  Note that (weak) equational
theories form a poset ordered by $\T\leq\T'$ iff
$\T_\Sigma\leq \T'_\Sigma$ for all $\Sigma\in \Nom_{\ofs}$, where
$\T_\Sigma \leq \T'_\Sigma$ holds iff for every $e' \in \T'_\Sigma$
there exists an $e \in \T_\Sigma$ with $e \leq e'$. Similarly,
(weak) pseudovarieties of nominal monoids form a poset w.r.t.~the
inclusion ordering.
\begin{theorem}\label{thm:theories_vs_pseudovars}
  (Weak) equational theories and (weak) pseudovarieties of nominal monoids form
  dually isomorphic complete lattices.
\end{theorem}
The isomorphism maps a (weak) equational theory $\T$ to the (weak)
pseudovariety $\pvar(\T)$ of all orbit-finite monoids $M$ such that
each morphism $h\colon \Sigma^*\to M$ with $\Sigma\in \Nom_{\ofs}$
factorizes through some $e_\Sigma\in \T_\Sigma$. The inverse
maps a (weak) pseudovariety $\pvar$ to the (weak) equational theory
$\T(\pvar)$ where $\T(\pvar)_\Sigma$ consists of all quotients
$e\colon \Sigma^*\epito M$ with codomain $M\in \pvar$.

\subsection{The Nominal Eilenberg-Schützenberger Theorem}
\label{S:eilenschuetz}

In addition to their abstract category theoretic characterization in
\autoref{thm:theories_vs_pseudovars}, weak pseudovarieties of nominal
monoids admit an axiomatic description in terms of sequences of
equations, analogous to the classical result of Eilenberg and
Schützenberger~\cite{es76} for pseudovarieties of ordinary
monoids. The appropriate concept of equation is as follows:
\begin{defn}\label{def:equation}
  \begin{enumerate}
  \item An \emph{equation} is a pair
    $(s,t)\in X^*\times X^*$, denoted as $s=t$, where $X$ is an
    orbit-finite strong nominal set. A nominal monoid $M$
    \emph{satisfies} $s=t$ if for every equivariant map
    $h\colon X\to M$ one has $\widehat{h}(s)=\widehat{h}(t)$, where
    $\widehat{h}\colon X^*\to M$ denotes the unique extension of $h$
    to an equivariant monoid morphism.

  \item Given a sequence $E=(s_n=t_n)_{n \in \Nat}$ of equations
    (possibly taken over different orbit-finite strong nominal sets $X$ of
    generators), a nominal monoid $M$ \emph{eventually satisfies}
    $E$ if there exists an index $n_0 \in \Nat$ such that $M$ satisfies
    all the equations $s_n=t_n$ with $n\geq n_0$. We denote by
    $\pvar(E)$ the class of all orbit-finite nominal monoids that
    eventually satisfy $E$.
  \end{enumerate}
\end{defn}
\begin{rem}
  Equations can be presented syntactically as expressions of the form
  \begin{equation}\label{eq:nomeq}
    y_1: S_1,\ldots, y_n: S_n \vdash u=v,
  \end{equation}
  where $Y=\{y_1,\ldots, y_n\}$ is a finite set of variables,
  $S_1,\ldots, S_n\seq \At$ are finite sets of atoms, and $u,v$ are
  words in $(\Perm(\At)\times Y)^*$.  A nominal monoid $M$ is said to
  \emph{satisfy} \eqref{eq:nomeq} if for every \emph{variable
    interpretation}, i.e.~every map $h\colon Y\to M$ with
  $\supp_M(h(y_i))\seq S_i$ for $i=1,\ldots,n$, one has $\ol h(u)=\ol
  h(v)$. Here, $\ol h\colon (\Perm(\At)\times Y)^*\to M$ is the
  unique monoid morphism mapping $(\pi,y_i)$ to $\pi\o h(y_i)$. Every
  equation can be transformed into an expressively equivalent
  syntactic equation, and vice versa~\cite[Lemma~B.31]{mu19_arxiv}.
\end{rem}
\begin{theorem}[Nominal Eilenberg-Schützenberger Theorem]\label{thm:nomeilschuetz}
  A class $\pvar$ of orbit-finite no\-mi\-nal monoids forms a weak pseudovariety
  iff 
  $\pvar=\pvar(E)$ for some sequence $E$ of equations.
\end{theorem}

\begin{proof}[Proof sketch]
  The proof proceeds along the lines of the one for ordinary
  monoids~\cite{es76}, although some subtle modifications are
  required. The ``if'' direction is a routine verification. For the
  ``only if'' direction, let $\pvar$ be a
  weak pseudovariety. Using that there are only countably many orbit-finite
  monoids up to isomorphism, one can construct a sequence
  $M_0,M_1,M_2,\ldots$ of nominal monoids in $\pvar$ such that each
  $M\in \pvar$ is a quotient of all but finitely many $M_n$'s. Let
  $X_0,X_1,X_2,\ldots$ be the sequence of all (countably many) strong
  orbit-finite nominal sets up to isomorphism, and consider the
  equivariant congruence relation on $X_n^*$ given by
  \[
    s\equiv_n t
    \quad\text{iff}\quad
    \text{$M_n$ satisfies the equation $s=t$.}
  \]
  One then shows that the congruence $\equiv_n$ is generated by a
  finite subset $W_n\seq \mathord{\equiv_n}$, and that
  $\pvar=\pvar(E)$ for every sequence $E$ that lists all equations in
  the countable set $\bigcup_n W_n$.
\end{proof}

\begin{example}\label{ex:aperiodic}
  An orbit-finite nominal monoid $M$ is called
  \emph{aperiodic}~\cite{boj2013,cpl15} if there exists a natural
  number $n\geq 1$ such that $x^{n+1} = x^n$ for all $x\in M$. The class of all orbit-finite aperiodic nominal
  monoids forms a pseudovariety. Taking the set $Y=\{y\}$ of
  variables, it is not difficult to see that this pseudovariety is specified by the sequence of
  syntactic equations
  \[
    y: S_n \vdash y^{n+1}=y^n \qquad (n \in \Nat),
  \]
  where $S_n = \{a_0,a_1,\ldots, a_{n-1}\}$ is the set of the first
  $n$ atoms in the countably infinite set $\At=\{a_0,a_1,a_2\ldots\}$
  of all atoms, and we write $y$ for $(\id, y)\in \Perm(\At)\times Y$.
\end{example}

\section{Duality and the Nominal Eilenberg Correspondence}
\label{S:var}


%

In this section, we establish our nominal version of Eilenberg's
variety theorem. It is based on a dual interpretation of the concepts
of a (local) pseudovariety of nominal monoids and of an equational theory, introduced in
the previous section, under \emph{nominal Stone duality}.
\subsection{Nominal Stone Duality}
\label{S:duality}
A classical result from duality theory, known as \emph{discrete
  Stone duality}, states that the category of sets is dually equivalent to the
category of complete atomic boolean algebras, i.e.~complete boolean
algebras in which every non-zero element is above some atom. An
analogous duality holds for the category $\Nomfs$ of nominal sets and
finitely supported maps.
\begin{defn}
  A \emph{nominal complete atomic boolean algebra} (\emph{ncaba}) is a
  boolean algebra $(B,\vee,\wedge,\neg,\bot,\top)$ in $\Nom$ such that
  every finitely supported subset of $B$ has a supremum, and for every
  element $b\in B\setminus\{\bot\}$ there exists an atom (i.e.~a minimal element) $a\in B$ with
  $a\leq b$. Here, the partial order $\leq$ is defined as usual by $a\leq b$ iff $a\wedge b = a$. We denote by $\nCABAfs$ the category of ncabas and
  finitely supported morphisms (i.e.~finitely supported maps
  preserving all the boolean operations and suprema of finitely
  supported subsets), and by $\nCABA$ the (non-full) subcategory of
  ncabas and \emph{equivariant} morphisms.
\end{defn}
\begin{theorem}[Nominal Stone Duality]\label{thm:nominalstonedual}
  The categories $\nCABAfs$ and $\Nomfs$ are dually equivalent. The duality restricts to one between the subcategories $\nCABA$ and $\Nom$.
\end{theorem}
The restricted duality is due to Petri\c{s}an~\cite[Prop.~5.3.11]{petrisan12}.
\begin{rem}\label{rem:nom_vs_ncaba}
\begin{enumerate}
\item The equivalence functor $\Nomfs\xra{\simeq}\nCABAfs^{\op}$ maps a
  nominal set $X$ to the ncaba $\Pow X$ of finitely supported subsets
  of $X$ (equipped with the set-theoretic boolean operations), and a
  finitely supported map $f\colon X\to Y$ to the morphism
  $f^{-1}\colon \Pow Y\to \Pow X$ taking preimages. The inverse
  equivalence functor $\nCABAfs^{\op} \xra{\simeq} \Nomfs$ maps an ncaba $B$ to the
  equivariant subset $\mathsf{At}(B)$ of its atoms, with group action
  restricting the one of $B$.

\item The dual equivalence restricts to one between the full
  subcategories of orbit-finite nominal sets and \emph{atom-finite}
  ncabas, i.e.~ncabas whose set of atoms is orbit-finite. For
  atom-finite ncabas the property that every finitely supported subset has a supremum is equivalent to the
  weaker requirement that for every finite set $S\seq \At$, every
  {$S$-orbit} has a supremum. Indeed, given a finitely supported
  subset $X\seq B$ (say with finite support $S\seq \At$), put
  $X' := \{\,a\in \mathsf{At}(B): a\leq x\text{ for some $x\in
    X$}\,\}$.
Since $\leq$ is an equivariant relation, $X' \subseteq \mathsf{At}(B)$ is a subset with finite support
  $S$. Since $\mathsf{At}(B)$ is orbit-finite and thus has only finitely
  many $S$-orbits by \autoref{lem:s_orbit_finite}, we can express
  $X'$ as a finite union
  $X'=X_1'\cup\ldots\cup X_n'$ of $S$-orbits. Using that every element of $B$ is the join of the finitely supported set of all atoms below it, it follows that
  $\bigvee X = \bigvee X' = \bigvee X_1' \vee \ldots \vee \bigvee
  X_n'$, so $\bigvee X$ is a finite join of
  joins of $S$-orbits.
\end{enumerate}
\end{rem}

\subsection{Varieties of Data Languages}

For the notion of a \emph{language} over an alphabet $\Sigma\in \Nom$ and the corresponding concept of algebraic recognition by nominal monoids, there are two natural choices: consider equivariant subsets
$L\seq \Sigma^*$ and their recognition by equivariant monoid morphisms~\cite{bkl14,cpl15}, or consider finitely supported subsets $L\seq \Sigma^*$ and their recognition by finitely supported monoid morphisms \cite{boj2013}. 
For our duality-based approach to data languages,  it turns out that we need to work with an intermediate concept: finitely supported languages recognizable by equivariant monoid morphisms (see the discussion in  \autoref{rem:nom_vs_nomfs} below). That is, we work
with the following
\begin{defn}\label{D:reclang}
  A \emph{data language} over the alphabet $\Sigma\in \Nom$ is a
  finitely supported map $L\colon \Sigma^*\to 2$. It is
  \emph{recognized} by an equivariant monoid morphism
  $e\colon \Sigma^*\to M$ if there exists a finitely supported map
  $p\colon M\to 2$ with $L=p\o e$. In this case, we also say that $M$
  \emph{recognizes} $L$. A data language is \emph{recognizable} if it
  recognized by some orbit-finite nominal monoid.
\end{defn}
\begin{rem}\label{rem:discrete}
\begin{enumerate}
\item Identifying finitely supported maps into $2$ with finitely
  supported subsets, \autoref{D:reclang} can be restated: an
  equivariant monoid morphism $e\colon \Sigma^*\to M$
  \emph{recognizes} a language $L\seq\Sigma^*$ if there exists a finitely
  supported subset $P\seq M^*$ with $L=e^{-1}[P]$.
\item If $L$ is an equivariant recognizable language, then
  $p$ in \autoref{D:reclang} is also equivariant. Therefore, for 
  equivariant languages we recover the notion of recognition
  of~\cite{bkl14,cpl15}. 
  
\item If $\Sigma$ is a finite set (viewed as an orbit-finite discrete
  nominal set), a data language is just an ordinary formal language over
  the alphabet $\Sigma$. Indeed, the free nominal
  monoid $\Sigma^*$ is discrete, and thus every subset of $\Sigma^*$
  is finitely supported. Moreover, every orbit-finite nominal quotient
  monoid of $\Sigma^*$ is discrete and finite. Hence, the above notion
  of language recognition coincides with the classical recognition by
  finite monoids. In particular, for finite $\Sigma$, a recognizable
  data language is the same as a regular language.
\end{enumerate}
\end{rem}

\begin{example}\label{ex:reclanguages}
Examples of recognizable data languages over the alphabet $\Sigma=\At$ include 
(1)~every finite or cofinite subset $L\seq \At^*$ (see
\autoref{rem:localvar} below), (2)~$a\At^*$ for a fixed atom $a\in
\At$, and (3)~$\bigcup_{a\in \At} \At^*aa\At^*$.  The languages (4)~$\{a_1\ldots a_n\;:\; a_i\neq a_j \text{ for } i\neq j \}$, (5)~$\bigcup_{a\in \At} a\At^*a\At^*$, and (6)~$\At^*a\At^*$ for a fixed
atom $a\in \At$ are  not recognizable. The equivariant examples (3)--(5) are  taken from \cite{boj2013,bkl14}.
\end{example}
In previous work~\cite{ammu14} we have given a categorical account of
\emph{local varieties of regular languages} \cite{ggp08},  i.e.~sets of regular
languages over a fixed finite alphabet $\Sigma$ closed under the
set-theoretic boolean operations (finite union, finite intersection,
complement) and derivatives. This concept can be generalized to data languages. The \emph{derivatives} of a data language
$L\seq\Sigma^*$ with respect to a word $w\in \Sigma^*$ are given by
\[w^{-1}L = \{\,v\in \Sigma^*\;:\; wv\in L\,\} \quad\text{and}\quad Lw^{-1} = \{\,v\in \Sigma^*\;:\; vw\in L\,\}.\]
Since $\supp(w^{-1}L),\supp(Lw^{-1})\seq \supp(w)\cup \supp(L)$, the derivatives are again data languages.

\begin{defn}[Local Variety of Data Languages]\label{def:localvar}
  Let $\Sigma\in \Nom$. A \emph{local variety of data languages over
    $\Sigma$} is an equivariant set $\lvar_\Sigma\seq \Pow\Sigma^*$
  of recognizable data languages closed under the set-theoretic boolean
  operations, unions of $S$-orbits for every finite set $S\seq \At$ of atoms
  (that is, for every $L\in \lvar_\Sigma$ the language
  $\bigcup_{\pi \in \Perm_S(\At)} \pi\o L$ lies in $\lvar_\Sigma$),
  and derivatives.
\end{defn}
\begin{rem}\label{rem:localvar}
\begin{enumerate}
\item If $\Sigma$ is a finite set (viewed as a discrete nominal set), then by
  \autoref{rem:discrete} a local variety $\lvar_\Sigma$ consists of
  regular languages, and the closure under unions of $S$-orbits
  is trivial: since $\Pow\Sigma^*$ is discrete, every $S$-orbit has a
  single element. Thus, in this case, a local variety of data
  languages is precisely a local variety of regular languages.

\item However, in general the closedness under unions of $S$-orbits
  cannot be dropped, as it is neither trivial nor implied by the other
  conditions. To see this, consider the alphabet $\Sigma=\At$ and the
  equivariant set $\lvar_\At\seq \Pow\At^*$ of all finite or
  cofinite subsets of $\At^*$. Note that every finite language
  $L\seq \At^*$ is recognizable: let $n \geq 1$ be an upper bound on
  the length of words in $L$, and take the orbit-finite monoid
  $M=\At^{\leq n} \cup \{0\}$ consisting of all words over $\At$ of
  length at most $n$, and a zero element $0$. The multiplication
  $\bullet$ is defined as follows: given $v,w\in \At^{\leq n}$, if the
  word $vw$ has length at most $n$, put $v\bullet w = vw$. Otherwise,
  put $v\bullet w = 0$. Then the equivariant monoid morphism $e\colon \At^*\to M$ extending $a\mapsto a$
  recognizes $L$ since $L=e^{-1}[L]$. It follows that also
  $\At^*\setminus L = e^{-1}[M\setminus L]$. This shows that every
  language in $\lvar_\At$ is recognizable. Moreover, clearly
  $\lvar_\At$ is closed under the set-theoretic boolean operations
  and derivatives. However, the languages $\{a\}$, $a\in \At$, form an
  orbit in $\lvar_\At$, but their union
  $\At=\bigcup_{a\in \At} \{a\}$ is not in $\lvar_\At$. Thus
  $\lvar_\At$ is not a local variety of data languages in the sense
  of \autoref{def:localvar}.
\end{enumerate}
\end{rem}
A local variety $\lvar_\Sigma$ is generally not a subobject of $\Pow\Sigma^*$ in $\nCABA$, because it is not required to be closed under unions of arbitrary finitely supported subsets and also not necessarily atomic as a boolean algebra. However, if the atomic languages in $\lvar_\Sigma$ form an orbit-finite subset and every language in $\lvar_\Sigma$ contains some atomic language, then $\lvar_\Sigma$ is an atom-finite subobject of $\Pow\Sigma^*$, see \autoref{rem:nom_vs_ncaba}\ref{rem:nom_vs_nomfs:2}. In this case, we call $\lvar_\Sigma$ an \emph{atom-finite} local variety.
\begin{theorem}[Finite Local Variety Theorem]\label{thm:finitelocalvar}
  The lattice of atom-finite local varieties of data languages over
  $\Sigma$ is dually isomorphic to the lattice of $\Sigma$-generated 
  orbit-finite monoids. 
\end{theorem}
The isomorphism maps a $\Sigma$-generated orbit-finite monoid
$e\colon \Sigma^*\epito M$ to the atom-finite local variety of all
data languages recognized by $e$.
\begin{proof}
  By the duality of $\Nom$ and $\nCABA$, orbit-finite equivariant
  quotients $e\colon \Sigma^*\epito M$ of $\Sigma^*$ in $\Nom$ correspond
  bijectively to atom-finite subobjects
  $\lvar_\Sigma \monoto \Pow\Sigma^*$ in $\nCABA$, i.e.~atom-finite
  equivariant sets of languages closed under the set-theoretic boolean
  operations and unions of $S$-orbits for every finite $S\seq \At$.  By
  \autoref{prop:unarypres} and the dual equivalence of $\Nomfs$ and
  $\nCABAfs$, the map $e$ represents a nominal quotient monoid of
  $\Sigma^*$ if and only if $\lvar_\Sigma$ is closed under derivatives,
  i.e.~a local variety. The closure under left derivatives is
  illustrated by the two dual commutative squares below, where the
  left-hand one lives in $\Nomfs$ and the right-hand one in
  $\nCABAfs$.
  \[
    \xymatrix@C+1pc@R-0.5pc{
      \Sigma^* \ar@{->>}[d]_e \ar[r]^{w\o \dash} & \Sigma^* \ar@{->>}[d]^e \\
      M \ar@{-->}[r]_\exists & M
    }
    \qquad\qquad
    \xymatrix@C+1pc@R-0.5pc{
      \Pow\Sigma^* & \Pow\Sigma^* \ar[l]_{w^{-1}(\dash)} \\
      \lvar_\Sigma \ar@{ >->}[u]^\seq
      &
      \lvar_\Sigma \ar@{-->}[l]^{\exists} \ar@{ >->}[u]_\seq
    }
  \]
  The elements of $\lvar_\Sigma$ are precisely the languages recognized
  by $e$. Indeed, the former correspond to the morphisms
  $\one \to \lvar_\Sigma$ in $\nCABAfs$, where $\one$ is the free boolean
  algebra on one generator, the latter to the finitely supported maps
  $M\to 2$ in $\Nomfs$, and $\one$ and $2$ are dual objects.
\end{proof}
Recall that an \emph{ideal} in a poset is a downwards closed and
upwards directed subset. For the lattice of local varieties of data
languages over $\Sigma$ (ordered by inclusion), we obtain
\begin{lemma}\label{lem:localvar}
  The lattice of local varieties of data languages over $\Sigma$ is
  isomorphic to the lattice of ideals in the poset of atom-finite local varieties
  over $\Sigma$.
\end{lemma}
The isomorphism maps a local variety $\lvar_\Sigma$ to the ideal of all
atom-finite local subvarieties of $\lvar_\Sigma$. Its inverse maps an
ideal $\{\,\lvar_{\Sigma,i} : i\in I \,\}$ in the poset of atom-finite
local varieties to the local variety
$\bigcup_{i\in I} \lvar_{\Sigma,i}$. In order-theoretic terms, the above lemma states that local
varieties of data languages form the ideal completion of the poset of
atom-finite local varieties. Using \autoref{thm:finitelocalvar},
\autoref{lem:localvar}, and the fact that ideals are dual to filters, we
obtain

\begin{theorem}[Local Variety Theorem]\label{T:localvar}
For each $\Sigma\in \Nom$, the lattice of local varieties of data languages over $\Sigma$ is dually isomorphic to the lattice of local pseudovarieties of $\Sigma$-generated nominal monoids.
\end{theorem}

\begin{rem}\label{rem:nom_vs_nomfs}
\begin{enumerate}
\item Since data languages are morphisms $L\colon \Sigma^*\to 2$ in
  $\Nomfs$, the reader may wonder why we do not entirely work in this
  category and use monoids with finitely supported multiplication and
  finitely supported monoid morphisms (rather than the equivariant
  ones) for the recognition of languages. The reason lies on
  the dual side: in the proof of \autoref{thm:finitelocalvar}, we used
  that equivariant injective maps $\lvar_\Sigma \monoto \Pow\Sigma^*$
  can be uniquely identified with equivariant subsets of
  $\Pow\Sigma^*$. In contrast, finitely supported injective maps
  $\lvar_\Sigma\monoto \Pow\Sigma^*$ do not correspond to the finitely
  supported (or any other kind of) subsets of $\Pow\Sigma^*$.
  
\item\label{rem:nom_vs_nomfs:2} Similarly, we cannot restrict
  ourselves to the category $\Nom$ and only consider equivariant
  languages $L\seq \Sigma^*$ rather than finitely supported
  ones. Indeed, the Finite Local Variety Theorem then fails: the map
  sending a $\Sigma$-generated orbit-finite monoid
  $e\colon \Sigma^*\epito M$ to the set of equivariant languages it
  recognizes is no longer bijective. To see this, consider the nominal
  monoids $M=\At\cup\{1\}$ with $a\bullet b = a$ for $a,b\in \At$, and
  $N=\{0,1\}$ with $0\bullet 0 = 0\bullet 1 = 1\bullet 0 = 0$. Then
  the two surjective morphisms $e\colon \At^*\epito M$, extending
  $a\mapsto a$, and $f\colon \At^*\epito N$, extending $a\mapsto 0$,
  recognize the same equivariant languages, namely $\At^*$,
  $\At^*\setminus \{\epsilon\}$, $\{\epsilon\}$ and $\emptyset$.
\end{enumerate}
\end{rem}
In the following, we consider data languages whose alphabet $\Sigma$
is an orbit-finite strong nominal set (see \autoref{rem:complete}\ref{rem:complete:3}). By dualizing the concept of an
equational theory, we obtain
\begin{defn}[Variety of data languages]\label{def:varlang}
  A \emph{variety of data languages} is a family
  \[
    \lvar = (\,\lvar_\Sigma\seq \Pow\Sigma^*\,)_{\Sigma\in\Nom_{\ofs}}
  \]
  of local varieties of data languages with the following two properties:
  \begin{enumerate}
  \item \emph{Closedness under preimages.} For each equivariant monoid
    morphism $h\colon \Delta^*\to \Sigma^*$ with
    $\Sigma,\Delta\in \Nom_{\ofs}$ and each $L\in \lvar_\Sigma$, one has
    $h^{-1}[L]\in \lvar_\Delta$.
    
  \item\label{def:varlang:2} \emph{Completeness.} For each atom-finite
    local subvariety $\lvar'_\Sigma\monoto \lvar_\Sigma$, there exists an
    equivariant monoid morphism $h\colon \Sigma^*\to \Delta^*$ and an
    atom-finite local subvariety $\lvar'_\Delta\monoto \lvar_\Delta$
    such that
    \begin{enumerate}
    \item\label{def:varlang:2:1} the map $L\mapsto h^{-1}[L]$ defines
      a bijection between $\lvar'_\Delta$ and $\lvar'_\Sigma$, and
    \item\label{def:varlang:2:2} every atomic language
      $L\in \lvar'_\Delta$ contains a word $w\in \Delta^*$ with
      $\supp_{\Pow\Delta^*}(L) = \supp_{\Delta^*}(w)$.
  \end{enumerate}  
\end{enumerate}
\end{defn}
\begin{rem}
  Except for the completeness condition, the above concept is
  analogous to Eilenberg's original notion of a variety of regular
  languages (i.e.~a family of local varieties of regular languages
  closed under preimages of monoid morphisms). In fact, if $\Sigma$ is
  a finite alphabet, and thus $\lvar_\Sigma$ is just a local variety
  of regular languages, completeness is trivial: given any finite
  local subvariety $\lvar'_\Sigma$ of $\lvar_\Sigma$, choose
  $\Delta=\Sigma$, $\lvar'_\Delta = \lvar'_\Sigma$, and
  $h=\id\colon \Sigma^*\to \Delta^*$. Then \ref{def:varlang:2:1} is
  clear, and \ref{def:varlang:2:2} holds because each
  $L\in \Pow\Sigma^*$ and each $w\in \Sigma^*$ has empty support.

  In general, however, the completeness property cannot be
  dropped. This follows from \autoref{rem:complete}\ref{rem:complete:2} and the
  observation that the completeness of a variety dualizes to the
  completeness of the corresponding equational theory (see the proof
  of \autoref{T:var}).
\end{rem}
We are ready to state the main result of our paper:
\begin{theorem}[Nominal Eilenberg Theorem]\label{T:var}
  Varieties of data languages and pseudovarieties of nominal monoids
  form isomorphic complete lattices.
\end{theorem}
\vspace{-0.05cm}
The isomorphism maps a variety $\lvar$ of data languages to the pseudovariety $\pvar$ of all orbit-finite nominal monoids that recognize only languages from $\lvar$. Its inverse maps a pseudovariety $\pvar$ to the variety $\lvar$ of all data languages recognized by some monoid in $\pvar$.

\begin{proof}[Proof sketch]
  We observe that the concept of a variety is dual to that of an equational
  theory. Indeed, by the Local Variety \autoref{T:localvar}, a family
  $\lvar = (\,\lvar_\Sigma\seq \Pow\Sigma^*\,)_{\Sigma\in\Nom_{\ofs}}$ of
  local varie\-ties of data languages bijectively corresponds to a
  family
  $\T = (\,\T_\Sigma\seq \Sigma^*\mathord{\epidownarrow}
  \NomMon_{\of}\,)_{\Sigma\in \Nom_{\ofs}}$ of local pseudovarieties
  of nominal monoids. One then shows that (1) $\T$ is
  substitution-invariant if and only if $\lvar$ is closed under
  preimages, and (2) $\T$ is complete if and only if $\lvar$ is
  complete. In particular, $\T$ is a theory if and only if $\lvar$ is a
  variety of data languages. Since theories correspond to pseudovarieties by \autoref{thm:theories_vs_pseudovars}, this proves the theorem. 
\end{proof}

\section{Adding Expressivity: Regular Data Languages}\label{sec:regdata}
As recognizing structures for data languages, nominal monoids are of
limited expressivity; in particular, they are strictly weaker than
deterministic automata in the category of nominal
sets~\cite{bkl14,boj2013}. Therefore, we now show how to extend our
results for monoid-recognizable data languages and establish a local
variety theorem for languages accepted by nominal automata.
\begin{defn}
  Fix an input alphabet $\Sigma\in \Nom$. A \emph{nominal $\Sigma$-automaton} $A=(Q,\delta,q_0)$ consists of a
  nominal set $Q$ of states, an equivariant transition map
  $\delta\colon Q\times \Sigma\to Q$, and an initial state $q_0\in Q$
  with empty support. It is called \emph{orbit-finite} if $Q$ is
  orbit-finite. A \emph{morphism} between nominal automata
  $A=(Q,\delta,q_0)$ and $A'=(Q',\delta',q_0')$ is an equivariant map
  $h\colon Q\to Q'$ such that $\delta(h(q),a)=h(\delta(q,a))$ for all
  $q\in Q$ and $a\in \Sigma$, and $h(q_0)=q_0'$.
\end{defn}
The \emph{initial nominal $\Sigma$-automaton} is given by
$I=(\Sigma^*, \delta, \epsilon)$ with transition map
$\delta(w,a) = wa$ for $w\in \Sigma^*$ and $a\in \Sigma$. It is
characterized by the universal property that for every {nominal
  $\Sigma$-automaton} $A=(Q,\delta,q_0)$, there exists a unique morphism
$e_A\colon I\to A$, sending a word $w\in \Sigma^*$ to the state
reached from $q_0$ after reading $w$. The automaton $A$ is called
\emph{reachable} if $e_A$ is surjective. A data language
$L\seq\Sigma^*$ is \emph{accepted} by $A$ if there exists a finitely
supported subset $F\seq Q$ with $L=e_A^{-1}[F]$. This corresponds to
the usual notion of acceptance of an automaton with final states $F$:
the language $L$ consists of all words $w\in \Sigma^*$ such that $A$
reaches a state in $F$ after reading $w$. A data language
$L\seq\Sigma^*$ is called \emph{regular} if there exists an
orbit-finite nominal automaton accepting it. In analogy to \autoref{prop:unarypres}, we get
\begin{proposition}[Unary presentation for nominal automata]\label{prop:unarypresaut}
  For every surjective equivariant map $e\colon \Sigma^* \epito Q$, the
  following statements are equivalent:
  \begin{enumerate}
  \item there exists a nominal automaton $A=(Q,\delta,q_0)$ with
    states $Q$ such that $e=e_A$;
  \item the maps $\Sigma^*\xra{\dash\o w} \Sigma^*$ $(w\in \Sigma^*)$
    lift along $e$, i.e.~there exist (necessarly unique) maps
    $r_w\colon Q \to Q$ such that $e\o (\dash\o w) = r_w\o e$ for all $w\in \Sigma^*$.
    \enlargethispage{4pt}
  \end{enumerate}
\end{proposition}
Define a \emph{local pseudovariety of nominal $\Sigma$-automata}
to be a class $\pvar_\Sigma$ of orbit-finite reachable nominal
$\Sigma$-automata such that (1) $\pvar_\Sigma$ is closed under quotients (represented by surjective automata morphisms), and (2) for every pair $A,B\in \pvar_\Sigma$, the reachable
  part of the product $A\times B$ lies in $\pvar_\Sigma$. Here, the product of two nominal automata $A=(Q,\delta,q_0)$ and
$B=(Q',\delta',q_0')$ is given by
$A\times B = (Q\times Q', \delta, (q_0,q_0'))$ with
$\delta ((q,q'),a) = (\delta(q,a), \delta'(q',a))$ for
$(q,q')\in Q\times Q'$ and $a\in \Sigma$, and the reachable part $R$ of $A\times B$ is the coimage $e\colon \Sigma^*\epito R$ of the unique morphism $e_{A\times B}\colon \Sigma^*\to A\times B$. Note that a local
pseudovariety corresponds precisely to a filter in the poset
$\Sigma^*\mathord{\epidownarrow} \Sigma\text{-}\mathbf{nAut}_\of$
of orbit-finite reachable nominal $\Sigma$-automata. The dual version
of this concept is the one of a \emph{local variety of regular data
  languages over $\Sigma$}: an equivariant set $\lvar_\Sigma\seq\Pow\Sigma^*$ of
such languages closed under the set-theoretic boolean operations,
unions of $S$-orbits for every finite set $S\seq\At$ of atoms, and
\emph{right} derivatives. The following theorem, and its proof, are
completely analogous to \autoref{T:localvar}:
\begin{theorem}[Local Variety Theorem for Regular Data Languages]\label{thm:locvar_reg}
For each $\Sigma\in \Nom$, the lattice of local varieties of regular data languages over $\Sigma$ is dually isomorphic to the lattice of local pseudovarieties of nominal $\Sigma$-automata.
\end{theorem}

\section{Conclusions and Future Work}\label{S:future}

We have demonstrated that two cornerstones of the algebraic theory of
regular languages, Eilenberg's variety theorem and Eilenberg and
Schützenberger's axiomatic characterization of pseudovarieties, can be
generalized to data languages recognizable by orbit-finite
monoids. Our results are the first of this type for data languages,
and thus the present work makes a contribution
towards developing a fully fledged algebraic theory of such
languages. In a broader sense, the approach taken in this paper can be
seen as a further illustration of the power of duality in formal
language theory: we believe that without the guidance given by nominal
Stone duality, it would have been significantly harder to even come up
with the suitable notion of a variety of data languages that makes the
nominal Eilenberg correspondence work. The duality-based approach thus adds
much conceptual clarity and simplicity. There remain several 
research questions and interesting directions for future work.

As indicated in \autoref{sec:regdata}, the techniques
used in our paper can be
adapted without much effort to languages recognized by nominal
algebraic structures other than monoids, including deterministic
nominal automata. As a first step, we aim to extend the local variety
theorem for regular data languages (\autoref{thm:locvar_reg}) to a
full Eilenberg correspondence. It remains an important goal to further extend our results to more powerful classes of data
languages.

Our proof of the (local) Eilenberg correspondence rests on the
observation that a local variety of data languages can be expressed as
the directed union of its atom-finite subvarieties. From a category theoretic perspective, this suggests that
local varieties are formed within the \emph{$\mathsf{Ind}$-completion}
(i.e.~the free completion under directed colimits) of the category of
atom-finite nominal complete atomic boolean algebras. We conjecture that this
completion can be described as a category of nominal boolean algebras
with joins of $S$-orbits for each finite set $S$ of atoms. On the dual
side, we expect that the \emph{$\mathsf{Pro}$-completion} (i.e.~the
free completion under codirected limits) of the category of
orbit-finite nominal sets consists of some form of nominal
Stone spaces. The approach of working with free completions should
lead to a topological version of nominal Stone duality similar to the
one established by Gabbay, Litak, and
Petri\c{s}an~\cite{gabbay11}. More importantly, it might pave the way
to the introduction of pro-(orbit-)finite methods for the theory of data
languages.

\bibliographystyle{plainurl}
\bibliography{refs_final}

\clearpage
\appendix 

\section{Appendix: Omitted Proofs and Details}

In this appendix, we provide full proofs of all our results and
technical details omitted due to space restrictions.

\section*{Properties of nominal sets}
We review some additional concepts from the theory
of nominal sets that we shall use in subsequent proofs.

\begin{defn}
  A \emph{supported set} is a set $X$ together with a map
  $\supp_X\colon X\to \Pow_f\At$, where $\Pow_f \At$ is the set of finite subsets of $\At$. A \emph{morphism} between
  supported sets $X$ and $Y$ is a function $f\colon X\to Y$ with
  $\supp_Y(f(x))\seq \supp_X(x)$ for all $x\in X$. We denote by
  $\SuppSet$ the category of supported sets and their morphisms.
\end{defn}
Note that every nominal set $X$ is a supported set w.r.t.~its
least-support function $\supp_X$, and that every equivariant map is a
morphism of supported sets. The following result is a reformulation of
\cite[Prop.~5.10]{msw16}:
\begin{lemma}[$\!\!$\cite{mu19_arxiv}, Lemma~B.25]\label{lem:nomreflective}
  The forgetful functor from $\Nom$ to $\SuppSet$ has a left adjoint.
\end{lemma}
The left adjoint $F\colon \SuppSet\to \Nom$ is constructed as follows.
Given a finite set $I$, let $\At^I = \prod_{i\in I} \At$ denote the
$I$-fold power of $\At$, and consider the strong nominal
set \[\At^{\#I} = \{\,a\in \At^{I}\;:\; \text{$a\colon I \to \At$ injective} \,\},\]
with group action given by $(\pi\o a)(i) := \pi(a(i))$ for every
$\pi\in\Perm(\At)$. Then $F$ sends a supported set $X$ to the nominal
set $FX = \coprod_{x\in X} \At^{\#{\supp_X(x)}}$, and the universal
map $\eta_X\colon X\to FX$ maps an element $x\in X$ to the inclusion
map $\supp_X(x)\monoto \At$ in $\At^{\#\supp_X(x)}$.
\begin{lemma}[$\!\!$\cite{mu19_arxiv}, Lemma~B.27]\label{L:strongprops}
  \begin{enumerate}
  \item For each nominal set $Z$, there exists a strong nominal set $X$
    and a surjective equivariant map $e\colon X\epito Z$ preserving
    least supports, i.e.~with $\supp_Z(e(x))=\supp_X(x)$ for all
    $x\in X$.
  \item Every strong nominal set is isomorphic to $FY$ for some $Y\in \SuppSet$.
  \end{enumerate}
\end{lemma}
It follows from the proof that in the case where $Z$ above is
orbit-finite, one may choose $X$ to be orbit-finite (in fact, with the same number of orbits as $Z$). In particular:
\begin{corollary}[Pitts~\cite{pitts2013}, Exercise~5.1]\label{L:singleorbit}
  Every nominal set which has only a single orbit is a
  quotient of the nominal set $\At^{\# n} = \{(a_1,\ldots,a_n)\in
  \At^n \;:\; a_i\neq a_j \text{ for $i\neq j$}\}$ for some $n\geq 0$.
\end{corollary}
\begin{lemma}\label{L:finsupp}
  Every orbit-finite nominal set contains only finitely many elements
  of any given support.
\end{lemma}
\begin{proof}
  This is a consequence of the following facts:
  \begin{itemize}
  \item Every nominal set is the disjoint union of its orbits.
  \item Every single-orbit nominal set is a quotient of the
    nominal set $\At^{\# n}$ (see \autoref{L:singleorbit}).
  \item Every $\At^{\# n}$ satisfies the desired property and every
    equivariant map $f\colon X \to Y$ satisfies $\supp_Y(f(x)) \seq
    \supp_X(x)$ for every $x \in X$.\qedhere
  \end{itemize}
\end{proof}

\section*{Proof of \autoref{lem:s_orbit_finite}}

Clearly the nominal set $\At^{\#n}$ has only finitely many $S$-orbits:
two tuples $(a_1,\ldots,a_n)$  and $(b_1,\ldots,b_n)$  in $\At^{\#n}$ lie in the same $S$-orbit if and only if for
every $i=1,\ldots,n$ one has $a_i\in S$ iff $b_i\in S$ and in this
case $a_i=b_i$. Thus, by \autoref{L:singleorbit}, every nominal set
with a single orbit has finitely many $S$-orbits (using that
equivariant maps preserve $S$-orbits). Since every orbit-finite
nominal set is the finite coproduct of its orbits, this proves the
claim.
\qed

\section*{Proof of \autoref{prop:unarypres}}

The proof is analogous to the corresponding statement for ordinary monoids. Given a
nominal quotient monoid $e\colon \Sigma^*\epito (M,\bullet,1_M)$ as
in~\ref{prop:unarypres:1}, choose $l_w = e(w)\bullet \dash$ and
$r_w = \dash\bullet e(w)$.  Conversely, if~\ref{prop:unarypres:2}
holds, there is a unique monoid structure $(M,\bullet,1_M)$ on $M$
making $e$ a monoid morphism in $\Set$. It is uniquely defined by
\[ ev\bullet ew := e(vw) \text{ for $v,w\in \Sigma^*$} \quad\text{and}\quad 1_M:=e(\epsilon). \] 
Condition $(2)$ makes sure that $\bullet$  is well-defined, i.e. independent of the choice of $v$ and $w$. Moreover, the multiplication $\bullet$
is equivariant since, for every $\pi\in \Perm(\At)$,
\[
  \pi\o (ev\bullet ew)
  =
  \pi \o e(vw)
  =
  e(\pi\o (vw))
  =
  e((\pi v)(\pi w))
  =
  e(\pi v)\o e(\pi w) 
  =
  (\pi\o e(v))\bullet (\pi\o e(w)),
\]
using that $e$ and the monoid multiplication on $\Sigma^*$ are
equivariant.
\qed

\section*{Proof of \autoref{thm:theories_vs_pseudovars}}
This theorem can be derived as an instance of the Generalized Variety
Theorem established in our previous work~\cite[Theorem~3.15]{mu19}, which relates varieties of objects
in a category $\A$ with a categorical notion of equational theory. For
the convenience of the reader, we briefly recall this result. For the
statement of the theorem one fixes the following parameters:
\begin{itemize}
\item a category $\A$ with a proper factorization system $(\E,\M)$;
\item a full subcategory $\A_0\seq \A$;
\item a class $\Lambda$ of cardinal numbers;
\item a class $\X\seq \A$ of objects.
\end{itemize}
Recall that a factorization system $(\E,\M)$ is \emph{proper} if all morphisms in $\E$ are epic and all morphisms in $\M$ are monic. The idea is that $\A$ is some category of algebraic structures, $\A_0$
is the subcategory in which varieties are formed, $\Lambda$ specifies
the arities of products under which varieties are closed, and $\X$ is
the class of algebras over which equations are formed (thus, $\X$ is usually
some class of free algebras).  \emph{Quotients} and \emph{subobjects}
in $\A$ are taken w.r.t.~the classes $\E$ and $\M$. An object
$A\in \A$ is called \emph{$\X$-generated} if there exists a quotient
$e\colon X\epito A$ for some $X\in \X$. An object $X\in \A$ is
\emph{projective} w.r.t.~a morphism $e\colon A\to B$ if for every
morphism $f\colon X\to B$ there exists a morphism $g\colon X\to A$
with $e\o g = h$. We define the subclass $\E_\X\seq \E$ by
\[
  \E_\X
  = 
  \{\,e\in \E \;:\; \text{every $X \in \X$ is projective w.r.t.~$e$}\,\}.
\]
Our data is required to satisfy the following assumptions:
\begin{enumerate}[label=(A\arabic*)]
\item\label{A1} $\A$ has $\Lambda$-products, i.e.~for every $\lambda \in
  \Lambda$ and every family $(A_i)_{i < \lambda}$ of objects in $\A$, the product
  $\prod_{i< \lambda} A_i$ exists.
\item\label{A2} $\A_0$ is closed under isomorphisms, $\Lambda$-products and
  $\X$-generated subobjects. The last statement means that for every subobject $m\colon A \monoto B$ in $\M$ where
  $B\in \A_0$ and $A$ is $\X$-generated, one has $A \in \A_0$.
\item\label{A3} Every object of $\A_0$ is an $\E_\X$-quotient of some
  object of $\X$, that is, for every object $A \in \A_0$ there exists some $e\colon X \epito A$ in $\E_\X$ with domain $X\in \X$.
\end{enumerate} 

\begin{defn}\label{def:variety}
 A \emph{(weak) variety} is a full subcategory $\pvar\seq \A_0$ closed under
  ($\E_\X$-)quotients, subobjects, and $\Lambda$-products. More precisely,
  \begin{enumerate}
  \item for every ($\E_\X$-)quotient $e: A\epito B$ in $\A_0$ with
    $A \in \pvar$ one has $B\in \pvar$,
  \item for every $\M$-morphism $m: A \monoto B$ in $\A_0$ with $B \in \pvar$ one has $A \in \pvar$, and 
  \item for every family of objects $A_i$ ($i<\lambda$) in $\pvar$ with
    $\lambda\in \Lambda$ one has  $\prod_{i<\lambda} A_i \in \pvar$.
  \end{enumerate}
\end{defn}
Given an object $X\in \A$, we denote by $X\mathord{\epidownarrow} \A_0$ 
the poset of all quotients $e\colon X\epito A$ with codomain
$A\in \A_0$, ordered by $e\leq e'$ iff $e'$ factorizes through $e$. A
subset $\T_X\seq X\mathord{\epidownarrow} \A_0$ is called a
\emph{(weak) equation} if it is downwards $\Lambda$-directed, i.e.~every subset
of $S\seq\T_X$ with $\under{S}\in \Lambda$ has a lower bound in
$\T_X$, and closed under ($\E_\X$-)quotients, i.e.~for every quotient
$e\colon X\epito E$ in $\T_X$ and every ($\E_\X$-)quotient
$q\colon E\epito E'$ in $\A_0$, one has $q\o e\in \T_X$. An object $A\in \A_0$ \emph{satisfies} the (weak) equation $\T_X$ if every morphism $h\colon X\to A$ factorizes through some quotient in $\T_X$. 

\begin{defn}\label{def:eqtheory}
 A \emph{(weak) equational theory} is a family
  \[
    \mathscr{T}
    =
    (\,\mathscr{T}_X\seq X\mathord{\epidownarrow} \A_0\,)_{X\in \X}
  \] of (weak) equations with the following properties (illustrated by the diagrams below):
  \begin{enumerate}
  \item \emph{Substitution invariance.} For every morphism $h\colon X\to Y$
    with $X,Y\in\X$ and every $e_Y\colon Y\epito E_Y$ in $\mathscr{T}_Y$, the coimage
    $e_X\colon X\epito E_X$ of $e_Y\o h$ lies in $\mathscr{T}_X$.
  \item \emph{$\E_\X$-completeness.} For every $Y\in \X$ and every quotient $e\colon Y\epito E_Y$ in
    $\mathscr{T}_Y$, there exists an $X\in\X$ and a quotient $e_X\colon X\epito E_X$
    in $\mathscr{T}_X\cap \E_\X$ with $E_X=E_Y$.
  \end{enumerate}
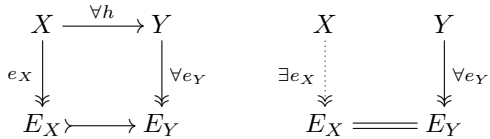
    \[
        \xymatrix{
          X \ar[r]^{\forall h} \ar@{->>}[d]_{ e_X} & Y 
          \ar@{->>}[d]^{\forall e_Y}\\
          E_X \ar@{>->}[r] & E_Y
        }
		\qquad \xymatrix{
          X \ar@{.>>}[d]_{\exists e_X} & Y
          \ar@{->>}[d]^{\forall e_Y}\\
          E_X\ar@{=}[r] & E_Y
        }
    \]
\end{defn}

\begin{rem}
We warn the reader about a clash of terminology: in our previous work  \cite{mu19_arxiv}, weak equations, weak equational theories and weak varieties were called  equations, equational theories and varieties (without the adjective ``weak''). Our present non-weak notion of an equation, an  equational theory, and variety in \autoref{def:variety}/\ref{def:eqtheory} was not considered in \cite{mu19_arxiv}.
\end{rem}

\begin{rem}
\begin{enumerate}
\item To every weak equational theory $\T$ one can associate the variety 
\[ \V(\T) = \{\, A\in \A_0 \;:\; \text{$A$ satisfies $\T_X$ for every $X\in \X$}\, \}. \]
Equivalently, by \cite[Lemma A.2]{mu19_arxiv} the variety $\V(\T)$ consists of all objects $A\in \A_0$ such that, for some $X\in \X$, there exists a quotient $e\colon X\epito A$ in $\T_X$ with codomain $A$.
\item Conversely, to every weak variety $\V$ one can associate a weak equational theory $\T(\V)$ where $[\T(\V)]_X$ consists of all quotients $e\colon X\epito A$ with $A\in \V$.
\end{enumerate}
\end{rem}
Given two weak equations $\T_X$ and $\T_X'$ over $X\in \X$, we put
  $\T_X\leq \T_X'$ if every quotient in $\T_X'$ factorizes through
  some quotient in $\T_X$. Weak theories form a poset with respect to the
   order $\T\leq \T'$ iff $\T_X\leq \T_X'$ for all $X\in
  \X$. Similarly, weak varieties form a poset  -- in fact, a complete lattice --
  ordered by inclusion. Then, we obtain the following correspondence: 
\begin{theorem}[$\!\!$\cite{mu19}, Theorem 3.15]\label{thm:hsp}
The complete lattices of weak equational theories and weak varieties are dually
  isomorphic. The isomorphism is given the maps
\[ \T\mapsto \V(\T)  \quad\text{and}\quad \V\mapsto \T(\V).  \]
\end{theorem}
Clearly, the above isomorphism restricts to one between equational theories and varieties. Thus, we obtain
\begin{corollary}\label{cor:stronghsp}
  The complete lattices of equational theories and varieties are dually
  isomorphic. 
\end{corollary}
Let us now instantiate these results to the setting of orbit-finite
nominal monoids. That is, we choose the parameters of the above categorical setting
as follows:
\begin{itemize}
\item $\A = \NomMon$;
\item $(\E,\M)$ = (surjective morphisms, injective morphisms);
\item $\A_0 = \NomMon_{\of}$;
\item $\Lambda =$ all finite cardinal numbers;
\item $\X = $ all free nominal monoids $\Sigma^*$ with $\Sigma$ an orbit-finite strong nominal set.
\end{itemize}
The class $\E_\X$ can be characterized as follows:
\begin{lemma}[$\!\!$\cite{mu19_arxiv}, Lemma B.28]\label{lem:supprefl}
A quotient $e\colon M\epito N$ belongs to $\E_\X$ if and only if $e$ is support-reflecting.
\end{lemma}
It is now easy to verify that the assumptions~\ref{A1}--\ref{A3} are
satisfied. For~\ref{A1} use that products of nominal monoids are
formed on the level of their underlying nominal sets. For~\ref{A2},
use that finite products and equivariant subsets of orbit-finite
nominal sets are again orbit-finite. For~\ref{A3}, let $M$ be an
orbit-finite nominal monoid. By \autoref{L:strongprops}, we can
express the underlying nominal set of $M$ as a quotient
$e\colon \Sigma\epito M$ in $\Nom$ of an orbit-finite strong nominal set
$\Sigma$, where $e$ preserves least supports. Then the unique extension
$\widehat{e}\colon \Sigma^* \epito M$ to an equivariant monoid morphism is
support-reflecting. Indeed, given $m\in M$, choose $a\in \Sigma$ with
$e(a)=m$. Then $\widehat{e}(a)=m$ and
$\supp_{\Sigma^*}(a) = \supp_\Sigma(a) = \supp_M(e(a))=\supp_M(m)$.

Finally, observe that the above general concepts of a variety and
of an equational theory specialize to the ones of a
pseudovariety of nominal monoids and of an equational theory of
nominal monoids, respectively. Thus, \autoref{thm:theories_vs_pseudovars} is an instance of
 \autoref{thm:hsp} and \autoref{cor:stronghsp}.
\qed

\section*{Proof of \autoref{thm:nomeilschuetz}}
The key ingredient of the proof is an observation about congruences. A \emph{congruence} on a nominal monoid $M$ is an equivariant equivalence relation $\mathord{\equiv}\seq M\times M$ such that, for $m,n\in M$,
\[m\equiv n \quad{\text{implies}}\quad pmq \equiv pnq \text{ for all $p,q\in M$}. \]
There is an isomorphism of complete lattices
\[ \text{congruences on $M$} \quad\cong\quad \text{nominal quotients monoids of $M$} \]
mapping a congruence $\equiv$ to the quotient monoid $e\colon M\epito
M/\mathord{\equiv}$, where $M/\mathord{\equiv}$ is the monoid of all
$\equiv$-congruence classes (with multiplication $[m]\bullet [n] =
[mn]$ for $m,n\in M$ and unit $[1_M]$), and $e$ sends $m$ to
$[m]$. The inverse isomorphism maps a nominal quotient monoid $e\colon M\epito N$
to its \emph{kernel}, i.e.~the congruence $\equiv_e$ on $M$ given by
$m\equiv_e m'$ iff $e(m)=e(m')$. As a consequence, we get
\begin{theorem}[Homomorphism Theorem]\label{T:hom}
  Given two equivariant monoid morphisms $e\colon M\epito N$ and
  $h\colon M\to P$ with $e$ surjective, there exists a morphism
  $g\colon N\to P$ with $g\o e = h$ if and only if the kernel of $e$
  is contained in the kernel of $h$.
\end{theorem}
A congruence  $\mathord{\equiv}$ on $M$ is is called \emph{orbit-finite} if the corresponding quotient monoid
$M/\mathord{\equiv}$ is orbit-finite, and \emph{finitely generated} if
there is a finite subset $W\seq \mathord{\equiv}$ such that
$\equiv$ is the least congruence on $M$ containing $W$.
\begin{lemma}\label{L:cong}
 Let 
  $X$ be an orbit-finite nominal set. Then every orbit-finite congruence
  on the free nominal monoid $X^*$ is finitely generated.
\end{lemma}

\begin{proof}
  Let $\equiv$ be an orbit-finite congruence on
  $X^*$.  Since $X^*/\mathord{\equiv}$ is orbit-finite, there exists a natural number $k>0$ such that each
  congruence class contains a word of length $<k$. Moreover, there
  exists a natural number $n$ such that every word of length $\leq k$
  in $X^*$ has a support of size $\leq n$. This follows from
  the fact that (a)~the least support of a word $a_1\cdots a_m\in \Sigma^*$ is given by $\supp_{X^*}(a_1\cdots a_m) = \bigcup_{i=1}^m
  \supp_X(a_i)$, and (b)~the size of least supports
  of elements of $X$  has an upper bound because $X$ is orbit-finite and for
  elements $x, x' \in X$ in the same orbit we have $|\supp_X(x)| =
  |\supp_X(x')|$.
  
  We now fix a list of $2n$ pairwise distinct atoms
  $a_1,\ldots, a_{2n}\in \At$, and consider the subset
  $W\seq \mathord{\equiv}$ defined by
  \[
    W := \{\, (s,t)\;:\; s\equiv t,\, \under{s}<k,\, \under{t}\leq
    k,\, \supp(s),\supp(t)\seq \{a_1,\ldots, a_{2n}\}  \,\}.
  \]
  Here $\under{\dash}$ denotes the length of a word. The set $W$ is
  finite because $X$ contains only finitely many elements with support $\{a_1,\ldots, a_{2n}\}$, see \autoref{L:finsupp}. Let
  $\sim$ be the congruence on $X^*$ generated by
  $W$ (i.e.~the intersection of all congruences containing $W$). We claim that $\mathord{\sim} = \mathord{\equiv}$, which proves
  that $\equiv$ is finitely generated.
  \begin{enumerate}
  \item\label{L:cong:1} For every pair $s,t\in X^*$ with $s\equiv t$,
    $\under{s}<k$ and $\under{t}\leq k$, one has $s\sim t$. Indeed,
    since $s$ and $t$ have supports of size $\leq n$, there exist
    pairwise distinct atoms $b_1,\ldots, b_{2n}$ such that
    $\supp(s),\supp(t)\seq \{b_1,\ldots, b_{2n}\}$. Choose a
    finite permutation
    $\pi\in \Perm(\At)$ with $\pi(b_i)=a_i$ for $i=1,\ldots, 2n$. Then the pair
    $(\pi\o s, \pi\o t)$ lies in $W$; indeed, we have
    \begin{enumerate}
    \item $\pi\o s \equiv \pi\o t$ because $s\equiv t$ and
      $\equiv$ is equivariant;
    \item $\under{\pi\o s}<k$ and $\under{\pi\o t}\leq k$
      because $\under{s}<k$ and $\under{t}\leq k$ and
      $\pi\o (\dash)$ preserves the length of words;
    \item
      $\supp(\pi\o s), \supp(\pi\o t)\seq \{a_1,\ldots,
      a_{2n}\}$ because $\supp(s),\supp(t)\seq \{b_1,\ldots,b_{2n}\}$.
    \end{enumerate}

    \noindent Since $(s,t)=\pi^{-1}\o (\pi\o s, \pi\o t)$ and $(\pi\o s, \pi\o t)\in W$, we
    conclude that the pair $(s,t)$ lies in the congruence
    generated by $W$, i.e.~$s\sim t$.
    
  \item\label{L:cong:2} We claim that for every word $s\in X^*$, there exists a
    word $t\in X^*$ of length $<k$ with $s\sim t$. The proof is
    by induction on $\under{s}$. If $\under{s}<k$, one may take
    $t=s$. Thus let $\under{s}\geq k$. Since $k>0$, we have that
    $s=xs'$ for some $x\in X$ and $s'\in X^*$. Since $\under{s'}<\under{s}$, we have that
    $s'\sim t'$ for some word $t'\in X^*$ of length $<k$ by
    induction. This implies $s=xs'\sim xt'$ because
    $\sim$ is a congruence. By the choice of the number $k$, there
    exists a word $t\in X^*$ of length $<k$ with
    $t\equiv xt'$. Since $\under{t}<k$ and
    $\under{xt'}\leq k$, we have $t\sim xt'$ by
    Part~\ref{L:cong:1}. Therefore, $s\sim xt'\sim t$, as required.
      
  \item We are ready two prove that $\mathord{\sim} = \mathord{\equiv}$. The inclusion
    $\seq$ is obvious. Conversely, suppose that $s\equiv t$. By
    Part~\ref{L:cong:2}, there are words $s'$ and $t'$ of length $<k$
    with $s\sim s'$ and $t\sim t'$. Since $\mathord{\sim}\seq \mathord{\equiv}$, we
    get $s'\equiv s\equiv t\equiv t'$. Part~\ref{L:cong:1} then shows
    that $s'\sim t'$. Thus, we conclude $s\sim s' \sim t'\sim t$.
    \qedhere
  \end{enumerate}
\end{proof}

\begin{rem}\label{rem:countable}
  As an important consequence of the previous lemma, we note that
  there are only countably many orbit-finite nominal monoids up to
  isomorphism. This is because (1)~every orbit-finite nominal monoid
  is a quotient of $X^*$ for some orbit-finite set $X$, (2)~up to
  isomorphism, there are only countably many orbit-finite sets $X$
  (this follows from \cite[Theorem 5.13]{pitts2013}; see
  also~\cite[Lemma~A.1]{amsw19}), and (3)~$X^*$ is countably infinite,
  and thus has only countably many orbit-finite quotient monoids by
  \autoref{L:cong} and the coincidence between quotients and
  congruences.
\end{rem}
We are ready to prove \autoref{thm:nomeilschuetz}.

\subsection*{Proof of the ``if'' direction}
Let $E=(s_n=t_n)_{n\in \Nat}$ be a sequence of nominal equations, where $s_n,t_n\in X_n^*$ for a strong orbit-finite nominal set $X_n$. We need to show that $\pvar(E)$ forms a weak pseudovariety.

\medskip
\noindent \emph{Closedness under finite products.} Let
$M,N\in \pvar(E)$. Then there exists $n_0\in\Nat$ such that both $M$
and $N$ satisfy the equations $s_n=t_n$ for $n\geq n_0$. Let
$h\colon X_n\to M\times N$ be an equivariant map, and denote by
$\pi_M\colon M\times N\to M$ and $\pi_N\colon M\times N\to N$ the
projections. Then, for every $n\geq n_0$,
\[
  \pi_M\o \widehat{h}(s_n)
  =
  \widehat{\pi_M\o h}(s_n)
  =
  \widehat{\pi_M\o h}(t_n)
  =
  \pi_M\o \widehat{h}(t_n)
\]
and analogously $\pi_N\o \widehat{h}(s_n)=\pi_N\o \widehat{h}(t_n)$. This
implies $\widehat{h}(s_n)=\widehat{h}(t_n)$. Thus $M\times N$ satisfies
$s_n=t_n$ for $n\geq n_0$, whence~$M\times N\in \pvar(E)$.

\medskip
\noindent \emph{Closedness under submonoids.} Let $M\in \pvar(E)$ and
let $m\colon N\monoto M$ be a nominal submonoid of $M$. Choose
$n_0\in \Nat$ such that $M$ satisfies all equations $s_n=t_n$ with
$n\geq n_0$. Then, for each equivariant map $h\colon X_n\to N$ and 
$n\geq n_0$, one has 
\[
  m \cdot \widehat{h}(s_n)
  =
  \widehat{m\o h}(s_n)
  =
  \widehat{m\o h}(t_n)
  =
  m \cdot \widehat{h}(t_n).
\]
Since $m$ is monomorphic, this implies that $\widehat h (s_n) = \widehat h (t_n)$.
Thus, $N$ satisfies $s_n=t_n$ for $n\geq n_0$, whence $N\in \pvar(E)$.

\medskip
\noindent \emph{Closedness under quotient monoids.} Let
$M\in \pvar(E)$ and let $e\colon M\epito N$ be a support-reflecting
surjective equivariant monoid morphism. Choose $n_0\in \Nat$ such that $M$
satisfies all equations $s_n=t_n$ with $n\geq n_0$. Given an
equivariant map $h\colon X_n\to N$, since $e$ is support-reflecting,
there exists an equivariant map $g\colon X_n\to M$ with $h=e\o g$, see
\autoref{lem:supprefl}. It follows that, for each $n\geq n_0$,
\[
  \widehat{h}(s_n)
  =
  \widehat{e \cdot g}(s_n)
  =
  e\o \widehat{g}(s_n)
  =
  e\o \widehat{g}(t_n)
  =
  \widehat{e \cdot g}(t_n)
  =
  \widehat{h}(t_n).
\]
Thus $N$ satisfies $s_n=t_n$ for $n\geq n_0$, whence $N\in \pvar(E)$.
\qed

\subsection*{Proof of the ``only if'' direction}
Let $\pvar$ be a
  weak pseudovariety of nominal monoids. We need to find a sequence $E$ of equations with $\pvar=\pvar(E)$.
  \begin{enumerate}
  \item\label{T:es:1} First, we prove that there exists a sequence of monoids
    $M_0,M_1,M_2,\ldots \in \pvar$ such that every $M\in \pvar$ is a
    support-reflecting quotient of $M_n$ for all but finitely many
    $n$. To see this, let $S_0,S_1,S_2,\ldots $ be the countably many
    elements of $\pvar$ up to isomorphism (see
    \autoref{rem:countable}), and put
    \[
      M_n = S_0\times S_1\times \cdots \times S_n\quad \text{for
        $n\geq 0$}.
    \]
    For each $M\in \pvar$ one has $M\cong S_m$ for some $m\geq 0$. Then,
    for $n\geq m$, we obtain the surjective morphism $M_n \epito
    M$ given by the projection
    \[
      M_n = S_0\times S_1\times \cdots \times S_n
      \xra{\pi_{n,m}} 
      S_m \xra{\cong} M.
    \]
    Note that $\pi_{n,m}$ is support-reflecting: for each
    $s\in S_m$ one has
    \[
      \pi_{n,m}(1,\ldots, 1, s , 1,\ldots , 1) = s,
    \]
    and since the unit $1$ of each nominal monoid has empty
    support, one has
    $\supp_{M_n}(1,\ldots, 1, s , 1,\ldots , 1) =
    \supp_{S_m}{s}$. Thus, we have shown that $M$ is a support-reflecting quotient of
    $M_n$ for each $n\geq m$.

  \item\label{T:es:2} Let $X_0,X_1,X_2,\ldots$ be a sequence that enumerates all orbit-finite
    strong nominal sets (up to isomorphism);
    cf.~\autoref{rem:countable}. For each $n\geq 0$, consider the
    following congruence $\equiv_n$ on $X_n^*$:
    \[
      s\equiv_n t \quad\Lra\quad h(s)=h(t)\text{ for all morphisms
        $h\colon X_n^*\to M_n$}.
    \]
    In other words, $\equiv_n$ is the set of all equations over $X_n$ satisfied
    by $M_n$. We claim that
    \[
      X_n^*/\mathord{\equiv_n} \in \pvar
      \quad
      \text{for all $n \in \Nat$}.
    \]
    To see this, note first that there are only finitely many
    morphisms $X_n^*\to M_n$: such a morphism bijectively corresponds
    to an equivariant map $X_n\to M_n$, which in turn corresponds to a
    map of supported sets $Y_n\to M_n$, where $Y_n$ is a finite
    supported set with $X_n\cong F(Y_n)$ (see
    \autoref{lem:nomreflective} and \ref{L:strongprops}). Since $M_n$
    contains only finitely many elements of any given support by
    \autoref{L:finsupp}, there exist only finitely many such maps.

    Letting $h_0,\ldots, h_p$ denote the morphisms from $X_n^*$ to
    $M_n$, the induced morphism
    \[
      \langle h_i\rangle_{i\leq p} \colon X_n^*\to \prod_{i\leq p} M_n
    \]
    has kernel $\equiv_n$. Thus, we can conclude that
    $X_n^*/\mathord{\equiv_n}$ is a submonoid of the finite product
    $\prod_i M_n$, and therefore it lies in $\pvar$.

  \item Since $X^*_n/\mathord{\equiv_n} \in \pvar$ by part~\ref{T:es:2}
    above, $\equiv_n$ is an orbit-finite congruence on $X_n^*$ for
    each $n\geq 0$. Therefore, by \autoref{L:cong}, there exists a
    finite generating subset $W_n$ of $\equiv_n$. We may assume that
    $W_n$ is nonempty, so that $\bigcup_{n} W_n$ is a countably
    infinite set of equations. Let $E=(s_k=t_k)_{k\in \Nat}$ be a sequence that
    enumerates all equations in $\bigcup_n W_n$. We claim that
    $\pvar=\pvar(E)$.

    \smallskip To prove ``$\seq$'', let $M\in \pvar$. By Part (1), there exists $m\geq 0$ such that
    $M$ is a support-reflecting quotient of $M_n$ for each $n\geq m$. Since
    $M_n$ satisfies all equations in $W_n$, it follows that also $M$
    satisfies them for $n\geq m$ (using that satisfaction of equations
    is preserved by support-reflecting quotients, as shown in the ``only if'' direction of our proof). Thus $M$ satisfies
    all equations in $\bigcup_{n\geq m} W_n$, i.e.~it eventually
    satisfies $E$.

    \smallskip For the proof of ``$\supseteq$'', let $M\in
    \pvar(E)$. Note that there are infinitely many $n\in \Nat$ such
    that a support-reflecting quotient $h\colon X_n^*\epito M$ exists
    (indeed, one may take any orbit-finite strong nominal set $X_n$ of
    the form $\At^{\#k_0} + \cdots + \At^{\#k_{m-1}} + \At^{\#k}$,
    where $m$ is the number of orbits of $M$, $k_i$ is the size of the
    least support of any element of the $i^\text{th}$ orbit, and $k$ is any
    natural number with $k\geq k_0$). In particular, since $M$
    satisfies all but finitely many of the equations in $E$, one can
    choose $n\in\Nat$ such that (a)~$M$ satisfies the equations in
    $W_n$, and (b)~a support-reflecting quotient
    $h\colon X_n^*\epito M$ exists. Then
    \[
      h(s)=h(t) \quad\text{for all $(s,t)\in W_n$}, 
    \]
    that is, $W_n$ is contained in the kernel of $h$. Since $W_n$
    generates $\equiv_n$, it follows that also $\equiv_n$ is contained
    in the kernel of $h$. The Homomorphism \autoref{T:hom} then yields a
    morphism $e\colon X_n^*/\mathord{\equiv_n}\epito M$ with
    $h=e\o q$, where $q$ is the quotient
    corresponding to $\equiv_n$:
    \[
      \xymatrix{
        X_n^* \ar@{->>}[dr]_h \ar@{->>}[r]^-{q}
        & X_n^*/\mathord{\equiv_n} \ar@{->>}[d]^e \\
        & M
      }
    \]
    Note that $e$ is support-reflecting because $h$ is. Since
    $X_n^*/\mathord{\equiv_n}\in \pvar$ by part~\ref{T:es:2} and $\pvar$ is
    closed under support-reflecting quotients, we conclude that
    $M\in \pvar$.\qed
\end{enumerate}

\section*{Details for \autoref{ex:aperiodic}}
 We show that an orbit-finite nominal monoid $M$ is aperiodic if and only if it eventually satisfies the equations  
  \begin{equation}\label{eq:seqeq}
    y: S_n \vdash y^{n+1}=y^n \qquad (n \in \Nat),
  \end{equation}
  where $S_n = \{a_0,a_1,\ldots, a_{n-1}\}$ is the set of the first
  $n$ atoms in the countably infinite set $\At=\{a_0,a_1,a_2\ldots\}$
  of all atoms, and we write $y$ for $(\id, y)\in \Perm(\At)\times Y$.

For the ``if'' direction, suppose that $M$ is aperiodic, i.e. there exists a natural number $m\geq 1$ such that $x^{m+1}=x^m$ for every $x\in M$. This implies that $x^{n+1}=x^n$ for all $n\geq m$ and $x\in M$. In particular, this holds for all elements $x$ with support $S_n$, i.e. \eqref{eq:seqeq} is satisfied for all $n\geq m$.  

For the ``only if'' direction, suppose that $M$ eventually satisfies the equations \eqref{eq:seqeq}. Choose $n\in \Nat$ such that $M$ satisfies the  equation $y: S_n \vdash y^{n+1}=y^n$ and every element of $M$ has a support of size $n$. Given any element $x\in M$, choose a finite permutation $\pi\in\Perm(\At)$ that maps a support of $x$ to $S_n$. The $\pi\o x$ is as element with support $S_n$, and thus $(\pi\o x)^{n+1}=(\pi\o x)^n$ holds. This implies $x^{n+1}=x^n$ because the monoid multiplication on $M$ is equivariant. Thus, $M$ is aperiodic.

\section*{Proof of \autoref{thm:nominalstonedual}}
Recall that $\Nomfs$ denotes the category of nominal sets and finitely
supported maps and $\nCABAfs$ the category of ncabas and finitely
supported boolean homomorphisms preserving suprema of finitely
supported subsets. We know from
Petri\c{s}an~\cite[Prop.~5.3.11]{petrisan12} that the categories
$\nCABA$ and $\Nom$ are dually equivalent with the equivalence
functors being the restrictions of the ones in
\autoref{rem:nom_vs_ncaba}. Thus, in order to obtain the desired
result we prove the following
\begin{proposition}\label{P:dualityfs}
  The dual equivalence $\Nom\xra{\simeq}\nCABA^{\op}$ extends to a
  dual equivalence $\Nomfs^{\op}\xra{\simeq} \nCABAfs$ given by
  $X\mapsto \Pow X$ and $f\mapsto f^{-1}$.
\end{proposition}
\begin{proof}
  We first show that the above functor $\Nomfs^{\op}\to \nCABAfs$ is
  well-defined, that is, for every finitely supported map
  $f\colon Y\to X$ the map $f^{-1}\colon \Pow X \to \Pow Y$ is also
  finitely supported. Let $S\seq \At$ be a finite support of $f$. Then
  for every $X_0\in \Pow X$ and $\pi\in \Perm_S(\At)$ one has
  $\pi\o f^{-1}[X_0] = f^{-1}[\pi\o X_0]$ since, for all $y\in Y$,
  \begin{align*}
    y\in \pi\o f^{-1}[X_0] & \Lra \pi^{-1}\o y \in f^{-1}[X_0]\\
    & \Lra f(\pi^{-1}\o y)\in X_0 \\
    & \Lra \pi^{-1}\o f(y)\in X_0\\
    & \Lra f(y) \in \pi\o X_0\\
    &\Lra y\in f^{-1}[\pi\o X_0]
  \end{align*} 
  This proves that $f^{-1}$ has support $S$.

  It remains to verify that the functor $\Nomfs^{\op}\to \nCABAfs$ is (1) essentially surjective
  on objects, (2) faithful and (3) full. Point (1) is clear because on
  objects the functor coincides with the equivalence functor
  $\Nom\xra{\simeq}\nCABA^{\op}$. For (2), suppose that
  $f\neq g\colon Y\to X$ are two distinct finitely supported
  maps. Choose $y\in Y$ with $f(y)\neq g(y)$. Since the singleton $\set{f(y)}$ is
  finitely supported, we may apply $f^{-1}$ and $g^{-1}$ to it, and we
  clearly have $f^{-1}\set{f(y)} \neq g^{-1}\set{f(y)}$. Thus $f^{-1}\neq g^{-1}$. For
  (3), let $g\colon \Pow X\to \Pow Y$ be a morphism in $\nCABAfs$ with
  finite support $S\seq \At$. Then the sets $g(\{x\})\seq Y$
  ($x\in X$) form a partition of $Y$ since $g$ is a homomorphism of
  boolean algebras, i.e.~for every $y\in Y$ there exists a unique
  $x_y\in X$ with $y\in g(\{x_y\})$. Consider the map $f\colon Y\to X$ defined by
  $f(y):= x_y$. We claim that $f$ has support $S$. To see this, let
  $\pi \in \Perm_s(\At)$. Then we have
  $\pi\cdot g(\{x\}) = g(\{\pi \cdot x\})$ for every $x \in X$. Since
  $f(y) = x_y$ where $y \in g(\{x_y\})$ we know
  $\pi \cdot y \in \pi \cdot g(\{x_y\}) = g(\{\pi \cdot x_y\})$,
  hence, $f(\pi \cdot y) = \pi \cdot x_y = \pi \cdot f(y)$ as desired.

Finally, we show that $g = f^{-1}$. By the definition of $f$ the two
maps agree on singleton sets $\{x\}$. It follows that they agree on
all finitely supported subsets of $X$: every such set can be expressed
as the finitely supported union of its singleton subsets, and both $g$
and $f^{-1}$ preserve that union.
\end{proof}

\section*{Details for \autoref{ex:reclanguages}}

(1)~See \autoref{rem:localvar}.

\smallskip\noindent
(2)~We prove that the language $L=a\At^*$ (for a fixed atom $a\in \At$) is recognizable. Consider the orbit-finite monoid $M=\{1\}+\At$ with multiplication given by
\[ b\bullet c = b \quad\text{for all $b,c\in \At$}. \]
Let $e\colon \At^*\to M$ be the unique morphism given by $b\mapsto b$. Thus $e$ maps $\epsilon$ to $1$ and every word $bw$ ($w\in \Sigma^*$, $b\in \At$) to $b$. It follows that
\[ L=e^{-1}[\{a\}]. \]
Since $\{a\}$ is a finitely supported subset of $\At$, this shows that $e$ recognizes $L$.

\smallskip\noindent(3)--(5) are equivariant languages and were treated
in~\cite{bkl14,cpl15}.

\smallskip\noindent
(6)~We prove that the language $L=\At^*a\At^*$ (for a fixed atom $a\in \At$) is not recognizable. Suppose for the contrary that there exists an equivariant monoid morphism $e\colon \At^*\to M$ into an orbit-finite monoid $M$ that recognizes $L$. We claim that 
\[ \supp(e(w))=\supp(w)\quad \text{for every $w\in \At^*\setminus L$}. \]
Since $\At^*$ contains words of arbitrarily large finite support, this implies that also $M$ contains elements of arbitrarily large finite support, contradicting the assumption that $M$ is orbit-finite. 

To prove the above equation, note that $\supp(e(w))\seq \supp(w)$
holds by equivariance of $e$. For the reverse inclusion, let
$w\in \At^*\setminus L$ and $b\in \supp(w)$, i.e.~$w$ contains the
letter $b$ but not the letter $a$. Suppose for the sake of
contradiction that $b\not\in \supp(e(w))$. Then
$a,b\not\in \supp(e(w))$ and thus
\[ e(w) = (a\, b)\o e(w) = e((a\, b)\o w), \] where
$(a\, b)\in \Perm(\At)$ denotes the permutation that swaps $a$ and $b$
and leaves all other atoms fixed. But we have $w\not\in L$ and
$(a\, b)\o w\in L$, which implies that $e(w) \neq e((a\, b)\o w)$.
\section*{Proof of \autoref{lem:localvar}}
For a given set $W\seq \Pow\Sigma^*$ of recognizable languages we
denote by $\langle W\rangle\seq \Pow\Sigma^*$ the subobject of
$\Pow\Sigma^*$ in $\nCABA$ generated by $W$, obtained by closing $W$
under the group action, boolean operations, and unions of finitely
supported subsets. Note that if $W$ is closed under derivatives,
then so is $\langle W \rangle$, i.e.~$\langle W \rangle$ is then a local
variety. This follows from the following identities:
\begin{align*}
 w^{-1}(K\cup L) &= w^{-1}K\cup w^{-1}L, \\
 w^{-1}(K\cap L) &= w^{-1}K\cap w^{-1}L, \\
 w^{-1}(\Sigma^*\setminus L)   &= \Sigma^*\setminus w^{-1}L, \\
 w^{-1}(\bigcup_{j\in J} L_j) &= \bigcup_{j\in J} L_j, \\
 w^{-1}(\pi \o L) &= \pi\o ((\pi^{-1}\o w)L),
\end{align*}
for all $w\in \Sigma^*$, languages $K,L,L_j\seq \Sigma^*$ and $\pi\in \Perm(\At)$.

\begin{enumerate}
\item Given an ideal of atom-finite local varieties
  $\lvar_{\Sigma,i}\seq \Pow\Sigma^*$ ($i\in I$), the directed union
  $\lvar_\Sigma = \bigcup_i \lvar_{\Sigma,i}$ is a local
  variety. Indeed, to see that $\lvar_\Sigma$ is closed under unions
  of $S$-orbits for every finite $S\seq \At$, suppose that
  $W\seq \lvar_\Sigma$ is an $S$-orbit. Then $W \seq \lvar_{\Sigma,i}$
  for some $i$ because every $\lvar_{\Sigma,i}$ is an equivariant
  subset of $\Pow\Sigma^*$. Since $\lvar_{\Sigma,i}$ is closed under
  unions $S$-orbits, we get
  $\bigcup W\in \lvar_{\Sigma,i}\seq \lvar_\Sigma$. To show that
  $\lvar_\Sigma$ is closed under finite union, let
  $K,L\in \lvar_\Sigma$. Then, since the $\lvar_{\Sigma,i}$ form a
  directed set, there exists $i\in I$ with $K,L\in \lvar_{\Sigma,i}$
  Thus $K\cup L \in \lvar_{\Sigma,i}\seq \lvar_\Sigma$ because
  $\lvar_{\Sigma,i}$ is closed under finite union. That $\lvar_\Sigma$
  is closed under the remaining boolean operations, derivatives, and
  the group action is shown in the same way.
\item Given a local variety $\lvar_\Sigma\seq\Pow\Sigma^*$, form the
  family $\lvar_{\Sigma,i}\seq \Pow\Sigma^*$ ($i\in I$) of all
  atom-finite~local subvarieties of $\lvar_\Sigma$. We claim that this
  family forms an ideal. Clearly, it is downwards closed. To show that
  it is upwards directed, it suffices to verify that for every
  $i,j\in I$ the local variety
  $\langle \lvar_{\Sigma,i}\cup \lvar_{\Sigma,j}\rangle$ is
  atom-finite. Let $e_i\colon \Sigma^*\epito M_i$ and
  $e_j\colon \Sigma^*\epito M_j$ be the $\Sigma$-generated
  orbit-finite quotient monoids corresponding to $\lvar_{\Sigma,i}$
  and $\lvar_{\Sigma,j}$ by the Finite Local Variety
  \autoref{thm:finitelocalvar}. Form their subdirect product, i.e.~the
  coimage $e\colon \Sigma^*\epito M$ of the map
  $\langle e_i,e_j\rangle\colon \Sigma^*\to M_i\times M_j$. Then $e$
  is the infimum of $e_i$ and $e_j$ in the lattice of quotients (in
  $\Nom$) of $\Sigma^*$. Thus $e$ dualizes to
  $\langle \lvar_{\Sigma,i}\cup \lvar_{\Sigma,j}\rangle$, the supremum
  of $\lvar_{\Sigma,i}$ and $\lvar_{\Sigma,j}$ in the lattice of
  subobjects of $\Pow\Sigma^*$ in $\nCABA$. Moreover, since $M$ is
  orbit-finite, $\langle \lvar_{\Sigma,i}\cup \lvar_{\Sigma,j}\rangle$
  is atom-finite.
  
\item It remains to show that the constructions of (1) and (2) are
  mutually inverse. First, given an ideal of atom-finite local
  varieties $\lvar_{\Sigma,i}\seq \Pow\Sigma^*$ ($i\in I$) and an
  atom-finite local subvariety $\mathcal{W}_\Sigma\seq \bigcup_i \lvar_{\Sigma,i}$
  we need to prove that $\mathcal{W}_\Sigma=\lvar_{\Sigma,i}$ for some $i$. Let
  $L_1,\ldots, L_n$ be the finitely many atoms of $\mathcal{W}_\Sigma$ up to
  application of the group action. By the ideal property, we may
  assume that $L_1,\ldots,L_n\in \lvar_{\Sigma,j}$ for some $j\in I$. Since every element $L$ of the ncaba $\mathcal{W}_\Sigma$ can be expressed as the union of all atoms below it, and the set of these atoms is a finite union of $S$-orbits for some finite $S\seq\At$ (cf. \autoref{rem:nom_vs_ncaba}), one has  
  $\mathcal{W}_\Sigma = \langle L_1,\ldots,L_n\rangle \seq \lvar_{\Sigma,j}$ and
  thus $\mathcal{W}_\Sigma = \lvar_{\Sigma,i}$ for some $i$ by down-closure.
\item Conversely, let $\lvar_\Sigma\seq\Pow\Sigma^*$ be a local
  variety and let $\lvar_{\Sigma,i}$ $(i\in I)$ be the ideal of its
  atom-finite local subvarieties. Clearly, we have
  $\bigcup_i \lvar_{\Sigma, i} \subseteq \lvar_\Sigma$, so we only
  need to show $\lvar_\Sigma \seq \bigcup_i \lvar_{\Sigma,i}$.  For
  every $L\in \lvar_\Sigma$, since $L$ is recognizable, there exists
  an orbit-finite quotient monoid $e\colon \Sigma^*\epito M$
  recognizing $L$. By the Finite Local Variety
  \autoref{thm:finitelocalvar}, the quotient $e$ dualizes to an atom-finite local
  variety $\mathcal{W}_\Sigma\seq \Pow\Sigma^*$ containing $L$. Then
  $\lvar_{\Sigma,i}:= \mathcal{W}_\Sigma\cap \lvar_\Sigma$ is an atom-finite
  local subvariety of $\lvar_\Sigma$ containing $L$,
  i.e.~$L\in \lvar_{\Sigma,i}\seq \bigcup_i \lvar_{\Sigma,i}$.\qed
\end{enumerate}
\section*{Proof of \autoref{T:var}}

The proof boils down to the observation that varieties can be
interpreted as the dual concept of an equational theory. By the Local
Variety Theorem, a collection
$\lvar = (\,\lvar_\Sigma\seq \Pow\Sigma^*\,)_{\Sigma\in\Nom_{\ofs}}$
of local varieties of data languages corresponds bijectively to a
collection
$\T = (\,\T_X\seq X^*\mathord{\epidownarrow} \NomMon_{\of}\,)_{X\in
  \Nom_{\ofs}}$ of local pseudovarieties of nominal monoids. We claim
that
\begin{enumerate}
\item\label{T:var:1} $\T$ is substitution-invariant if and only if
  $\lvar$ is closed under preimages, and
\item\label{T:var:2} $\T$ is complete if and only if $\lvar$ is complete.
\end{enumerate}
In particular, $\T$ is an equational theory if and only if $\lvar$ is variety. Together with \autoref{thm:theories_vs_pseudovars}, this proves the theorem.

\medskip \noindent \emph{Proof of~\ref{T:var:1}.} By the Finite Local
Variety \autoref{thm:finitelocalvar} and \autoref{lem:localvar}, the substitution
invariance of $\T$ dualizes to the statement that for every
atom-finite local subvariety $\lvar'_\Sigma$ of $\lvar_\Sigma$ and
every $h\colon \Delta^*\to \Sigma^*$, there exists an atom-finite
local subvariety $\lvar'_\Delta$ of $\lvar_\Delta$ that contains the
preimage $h^{-1}[L]$ for each $L\in \lvar'_\Sigma$.
\[
  \xymatrix@C+1pc{
    \Delta^* \ar@{->>}[d]_{e_\Delta} \ar[r]^{h} & \Sigma^* \ar@{->>}[d]^{e_\Sigma} \\
    M_\Delta \ar@{-->}[r]_\exists & M_\Sigma
  }
  \qquad\qquad
  \xymatrix@C+1pc{
    \Pow\Delta^* & \Pow\Sigma^* \ar[l]_{h^{-1}(\dash)} \\
    \lvar'_\Delta \ar@{ >->}[u]^\seq & \lvar'_\Sigma \ar@{-->}[l]^{\exists} \ar@{ >->}[u]_\seq
  }
\]
By \autoref{lem:localvar}, every language in $\lvar_\Sigma$ is
contained in some atom-finite subvariety $\lvar'_\Sigma$. Hence, the
commutativity of the right-hand diagram above implies that $\lvar$ is closed
under preimages. Conversely, suppose that $\lvar$ is closed under
preimages, and let $\lvar'_\Sigma$ be an atom-finite local subvariety
of $\lvar_\Sigma$. Thus, there are atomic languages
$L_1,\ldots,L_n\in \lvar'_\Sigma$ such that every other atomic
language in $\lvar'_\Sigma$ is of the form $\pi\o L_i$ for some
$\pi\in\Perm(\At)$ and $i\in\{1,\ldots,n\}$. By closure of $\lvar$
under preimages, one has $h^{-1}[L_i]\in \lvar_\Delta$ for all
$i$. Choose an atom-finite local subvariety $\lvar'_{\Delta,i}$ of
$\lvar'_\Delta$ containing $L_i$. Since the atom-finite local
subvarieties of $\lvar'_\Delta$ form an ideal (i.e.~are upwards
directed), we can choose an atom-finite local subvariety
$\lvar'_\Delta$ containing all $\lvar'_{\Delta,i}$'s. Thus
$h^{-1}[L_i]\in \lvar'_\Delta$ for all $i$. By
\autoref{rem:nom_vs_ncaba}, every language $L\in \lvar'_\Sigma$ can be
expressed as a finite union $L=K_0\cup\ldots\cup K_m$, where every $K_i$ is
the union of an $S$-orbit in $\lvar'_\Sigma$. Since the equivariant
map $h^{-1}$ maps $S$-orbits to $S$-orbits, and $\lvar'_\Delta$ is
closed under finite unions and unions of $S$-orbits, it follows that
$h^{-1}[L]=h^{-1}[K_1] \cup\ldots\cup h^{-1}[K_m]$ lies in
$\lvar'_\Delta$.

\medskip\noindent \emph{Proof of~\ref{T:var:2}.} Suppose that $\T$ is
complete, let $\lvar'_\Sigma$ be an atom-finite local subvariety of
$\lvar_\Sigma$, and let $e_\Sigma\colon \Sigma^*\epito M_\Sigma$ be
its dual quotient in $\T_\Sigma$. Then there exists
$\Delta\in \Nom_\ofs$ and a support-reflecting quotient
$e_\Delta\colon \Delta^*\epito M_\Delta$ in $\T_\Delta$ with
$M_\Delta\cong M_\Sigma$. By projectivity (see
\autoref{lem:supprefl}), we can choose an equivariant monoid morphism
$h\colon \Sigma^*\to \Delta^*$ with $e_\Delta\o h = e_\Sigma$. Letting
$\lvar'_\Delta$ denote the atom-finite local subvariety of
$\lvar_\Delta$ dual to $e_\Delta$, we thus obtain the two dual
commutative diagrams below:
\[
  \xymatrix@C+1pc{
    \Delta^* \ar@{->>}[d]_{e_\Delta}  & \Sigma^* \ar[l]_h \ar@{->>}[d]^{e_\Sigma} \\
    M_\Delta  & M_\Sigma \ar@{-->}[l]^\cong
}
\qquad\qquad
\xymatrix@C+1pc{
  \Pow\Delta^* \ar[r]^{h^{-1}(\dash)} & \Pow\Sigma^*  \\
  \lvar'_\Delta \ar@{ >->}[u]^\seq \ar@{-->}[r]_{\cong} & \lvar'_\Sigma  \ar@{ >->}[u]_\seq
}
\]
The right-hand diagram shows that the map $L\mapsto h^{-1}[L]$ from
$\lvar'_\Delta$ to $\lvar'_\Sigma$ is bijective. Moreover, by
\autoref{thm:finitelocalvar}, $M_\Delta$ is (up to isomorphism) carried by the set of atomic languages in $\lvar'_\Delta$, and $e_\Delta$ is the
morphism mapping every word $w\in \Delta^*$ to the unique atomic
language in $\lvar'_\Delta$ containing $w$.  That $e_\Delta$ is
support-reflecting thus means precisely that every atomic language
$L\in \lvar'_\Delta$ contains a word $w\in L$ with
$\supp_{\Pow\Delta^*}(L) = \supp_{\Delta^*}(w)$. This shows that
$\lvar$ is complete.

Conversely, suppose that $\lvar$ is complete, and let
$e_\Sigma\colon \Sigma^*\epito M_\Sigma$ be a quotient in
$\T_\Sigma$. Then, for its dual atom-finite local subvariety
$\lvar'_\Sigma$ of $\lvar_\Sigma$, we have that there exists a
morphism $h\colon \Sigma^*\to \Delta^*$ and an atom-finite local
subvariety $\lvar'_\Delta$ of $\lvar_\Delta$ satisfying the
conditions~\ref{def:varlang:2:1} and~\ref{def:varlang:2:2} from
\autoref{def:varlang}\ref{def:varlang:2}. Condition~\ref{def:varlang:2:1}
means precisely that we have the right-hand commutative square
above. Thus, letting $e_\Delta\colon \Delta^*\epito M_\Delta$
denote the quotient in $\T_\Delta$ that dually corresponds to
$\lvar'_\Delta$, we obtain the left-hand diagram above. Moreover, condition~\ref{def:varlang:2:2} states that
$e_\Delta$ is support-reflecting. Thus $\T$ is complete.
\qed

\end{document}

%% file: preamble_nom.tex
\usepackage{amsmath,mathrsfs}
\usepackage[bbgreekl]{mathbbol}
\usepackage{xspace}
\usepackage{dsfont}
\usepackage{chngcntr}
\usepackage{float}
\usepackage[all]{xy}
\SelectTips {cm}{}
\newdir{ >}{{}*!/-5pt/@{>}}
\newdir{ (}{{}*!/-5pt/@{(}}

\DeclareFontFamily{U}{bbold}{}\DeclareFontShape{U}{bbold}{m}{n}{%
  <4.25>bbold5<5>bbold5<6>bbold6<7>bbold7<8>bbold8<8.5>bbold9<9.5>bbold10<10>bbold10<12>bbold12}{}

\usepackage[notcite,notref]{showkeys}

\hypersetup{
    colorlinks=true,
    linkcolor=black,
    citecolor=black,
    filecolor=black,
    urlcolor=black,
    pdftitle={Varieties of Data Languages},
    pdfauthor={Stefan Milius, Henning urbat},
    pdfkeywords={Nominal sets, Stone duality, Algebraic language
      theory, Data languages},
    pdfduplex={DuplexFlipLongEdge},
}

\usepackage{seqsplit}
\usepackage{xstring}
\usepackage{xcolor}

\makeatletter
\renewcommand*\showkeyslabelformat[1]{%
\@ifundefined{hideNextShowKeysLabel}{%
\noexpandarg%
\StrSubstitute{#1}{ }{\textvisiblespace}[\TEMP]%
\parbox[t]{\marginparwidth}{\raggedright\normalfont\small\ttfamily\(\{\){\color{red!50!black}\expandafter\seqsplit\expandafter{\TEMP}}\(\}\)}%
}{}
}
\makeatother

\theoremstyle{plain}
\newtheorem*{theorem*}{Theorem}
\newtheorem*{lemma*}{Lemma}

\newtheorem*{nomeilthm}{Nominal Eilenberg Theorem}
\newtheorem*{nomeilschuetz}{Nominal Eilenberg-Schützenberger Theorem}

\theoremstyle{definition}
\newaliascnt{rem}{theorem}
\newtheorem{rem}[rem]{Remark}
\aliascntresetthe{rem}

\newaliascnt{defn}{theorem}
\newtheorem{defn}[defn]{Definition}
\aliascntresetthe{defn}

\theoremstyle{remark}

\usepackage[inline]{enumitem}
\setlist[enumerate,1]{label=(\arabic*),font=\normalfont,align=left,leftmargin=0pt,labelindent=0pt,listparindent=\parindent,labelwidth=0pt,itemindent=!,topsep=0pt,parsep=0pt,itemsep=0pt,start=1}
\setlist[enumerate,2]{label=(\alph*),font=\normalfont,labelindent=*,topsep=0pt,leftmargin=*,start=1}
\setlist[itemize]{labelindent=*,leftmargin=*,topsep=5pt,itemsep=3pt}
\setlist[description]{labelindent=*,leftmargin=*,itemindent=-1 em}

\newcommand{\pvar}{\mathscr{V}}
\newcommand{\lvar}{\mathcal{V}}

\def\A{\mathscr{A}}
\def\V{\VCat}

\def\epsilon{\varepsilon}

\renewcommand{\rho}{\varrho}
\def\ol{\overline}
\def\o{\cdot}

\newcommand{\of}{\mathsf{of}}
\newcommand{\ofs}{\mathsf{of,s}}

\newcommand{\set}[1]{\{#1\}}

\newcommand{\takeout}[1]{\empty}

\newcommand{\supp}{\mathop{\mathsf{supp}}}

\newcommand{\At}{\mathbb{A}}
\newcommand{\Perm}{\mathrm{Perm}}
\newcommand{\Nom}{\mathbf{Nom}}
\newcommand{\Nomfs}{\mathbf{Nom}_\mathsf{fs}}
\newcommand{\nCABA}{\mathbf{nCABA}}
\newcommand{\nCABAfs}{\mathbf{nCABA}_\mathsf{fs}}
\newcommand{\NomMon}{\mathbf{nMon}}
\newcommand{\SuppSet}{\mathbf{SuppSet}}

\renewcommand{\phi}{\varphi}

\newcommand{\xra}{\xrightarrow}

\newcommand{\Lra}{\Leftrightarrow}
\newcommand{\seq}{\subseteq}

\newcommand{\T}{\mathscr{T}}

\newcommand\epidownarrow{\mathrel{\rotatebox[origin=c]{90}{$\twoheadleftarrow$}}}

\newcommand{\Set}{\mathbf{Set}}
\newcommand{\E}{\mathcal{E}}

\newcommand{\id}{\mathit{id}}

\newcommand{\VCat}{\mathscr{V}}

\newcommand{\X}{{\mathscr{X}}}

\newcommand{\under}[1]{{|#1|}}

\newcommand{\epito}{\twoheadrightarrow}

\newcommand{\monoto}{\rightarrowtail}

\newcommand{\M}{\mathcal{M}}

\newcommand{\Pow}{\mathcal{P}}

\newcommand{\op}{\mathsf{op}}

\newcommand{\Nat}{\mathds{N}}

\newcommand{\one}{\mathds{1}}

\mathchardef\ordinarycolon\mathcode`\:
\mathcode`\:=\string"8000
\begingroup \catcode`\:=\active
  \gdef:{\mathrel{\mathop\ordinarycolon}}
\endgroup

\mathchardef\hyph="2D
\newcommand{\dash}{\mathord{-}}
